\documentclass[12pt]{article}
\usepackage[margin=0.9in]{geometry}

\usepackage[utf8]{inputenc}
\usepackage{bm}
\usepackage{bbm}
\usepackage{stmaryrd}
\usepackage{amsmath}
\usepackage{amsthm}
\usepackage{amssymb}
\usepackage{xcolor}
\usepackage{graphicx}
\usepackage{ulem}
\usepackage{url}
\usepackage{comment}
\usepackage{braket}
\usepackage{cancel}

\DeclareFontFamily{U}{mathx}{}
\DeclareFontShape{U}{mathx}{m}{n}{<-> mathx10}{}
\DeclareSymbolFont{mathx}{U}{mathx}{m}{n}
\DeclareMathAccent{\widehat}{0}{mathx}{"70}
\DeclareMathAccent{\widecheck}{0}{mathx}{"71}

\DeclareMathOperator*{\esssup}{ess\,sup}

\newcommand{\tr}{\mathrm{Tr}}

\newcommand{\h}{{\mathcal H}}

\renewcommand{\r}{{\rm R}}

\newcommand{\s}{{\rm S}}

\newcommand{\cx}{{\mathbb C}}
\newcommand{\rx}{{\mathbb R}}
\newcommand{\hh}{L^2(\rx^3, d^3k)}

\newcommand{\vx}{{\mathbf x}}

\newcommand{\vy}{{\mathbf y}}

\newcommand{\vw}{{\mathbf w}}

\newcommand{\vk}{{\mathbf k}}

\newcommand{\dg}{{\mathcal P}}

\newcommand{\ind}{\textbf{1}}
\newcommand{\T}{{\mathcal T}_2(\rx^{dN})}

\newcommand{\bbbone}{\mathchoice {\rm 1\mskip-4mu l} {\rm 1\mskip-4mu l}
{\rm 1\mskip-4.5mu l} {\rm 1\mskip-5mu l}}
\newtheorem{thm}{Theorem}
\newtheorem{prop}{Proposition}
\newtheorem*{prop*}{Proposition}
\newtheorem{lem}{Lemma}
\newtheorem{cor}{Corollary}

\newcounter{example}[section] 
\renewcommand{\theexample}{\thesection.\arabic{example}} 

\newenvironment{example}[1][]{%
  \refstepcounter{example}%
  \par\medskip
  \noindent\textbf{Example \theexample}%
  \ifx&#1&\else\ (\textit{#1})\fi.
  \rmfamily\quad
}{%
  \par\medskip
}

\usepackage{authblk}
\title{On spatial decoherence in many-body systems}
\author[1]{Stefano Marcantoni\footnote{stefano.marcantoni@gssi.it}}
\author[2]{Marco Merkli\footnote{merkli@mun.ca}}
\affil[1]{Mathematics Division 

Gran Sasso Science Institute

Viale Rendina 26-28, 67100 L'Aquila, Italy
\medskip
}
\affil[2]{Department of Mathematics and Statistics

Memorial University of Newfoundland

St.~John’s, NL, Canada A1C 5S7
}

\begin{document}

\maketitle

\begin{abstract}
We study the dynamics of a many-body quantum system strongly interacting with a bosonic reservoir. The coupling is given by a potential operator for the system and it is linear in the field. We show that in the limit of infinite coupling strength the system undergoes instantaneous spatial decoherence. To resolve the decoherence process in time we consider large coupling strengths $\lambda$ and short times scaling as $t\propto \lambda^{-\alpha}$, with $\alpha\ge0$, in the limit $\lambda\rightarrow\infty$.  On short time scales $\alpha>1$ the dynamics is trivial while on longer ones $0\le\alpha<1$ the decoherence is instantaneous. We show that the decoherence process is resolved exactly for $\alpha=1$, defining the fine-grained time scale $\tau=\lambda t$. We construct an approximate effective evolution map of the many-body system with controlled error estimates. Generically, the effective dynamics is markovian, but not given by a dynamical semigroup. As a physical application we give a rigorous description of the phenomenon of localization of macroscopic quantum objects in position space.
\end{abstract}

\section{Introduction and overview of  results}

The loss of coherence is a widespread phenomenon in quantum system which interact with environments. The physical implications of decoherence are vast. They are widely discussed in the theoretical and experimental literature and take an important role in the discussion on then foundations of quantum theory \cite{Joosetal,Haroche-Raimond, Zurek, SchlossBook, Schlosshauer, BO2003}. The dynamical process of decoherence is usually defined as the decrease of the off-diagonal elements of the system density matrix in the course of time. The Hilbert space basis in which decoherence takes place is determined by the interaction between the system and the environment and so is the degree and the speed of the decoherence. The vanishing of off-diagonals is also achieved by applying a non-selective von Neumann projective quantum measurement on the system. Such a measurement results in a density matrix with entirely deleted off-diagonal element while the diagonal stays unaffected, in the eigenbasis of the measured observable \cite{NielsenChuang, Wilde}. Traditionally, those measurements are described by an instantaneous action | a quantum channel acting on the density matrix. They do not involve the description of a measurement apparatus provoking the measurement as a dynamical process. It is possible, though, to derive the instantaneous decoherence from a macroscopic open system description: It was shown in \cite{Marcantoni-Merkli} that when a (finite dimensional) quantum system is coupled to an environment (scalar bosonic quantum field), then instantaneous decoherence occurs as a consequence of taking the `Zeno' limit of infinite coupling strength. Naturally, we then expect that large but finite coupling strengths would cause quick, but not instantaneous decoherence. The temporal resolution of this decoherence process for many-body systems in contact with environments  is the subject of the present work. We focus on system-environment interactions which give rise to {\it spatial decoherence} with respect to the `position basis'.
\medskip

{\bf Overview of the main results.} We present here our main results in an informal way before detailing the mathematically precise assumptions and theorems in the next section. An $N$-body quantum system in $d$ spatial dimensions, with a  Hamiltonian $H_\s$ on  $L^2(\rx^{dN},d\vx)\equiv L^2(\rx^{dN})$ in a state $\rho$, is in contact with a scalar bosonic quantum field in a state $\omega_\r$. The interacting Hamiltonian is given by (omit trivial tensor factors)
$$
H = H_\s + H_\r +\lambda G\otimes\varphi(g)
$$
where $H_\r$ is the generator of the free dynamics of the field, $\varphi(g)$ is the field operator satisfying the CCR $[\varphi(f),\varphi(g)]=i{\rm Im}\langle f,g\rangle$ for test functions $f,g\in L^2(\rx^3,d^3k)$, and where $g$ in the Hamiltonian is the called the form factor. This type of reservoir is a standard choice to model `noise' in open system theory. We assume that $\omega_\r$ is a regular state so that the field operators are well defined. $G$ is the operator of   multiplication by a real valued function $G(\vx)$ acting on $L^2(\rx^{dN})$. The $\lambda\in\rx$ is a coupling constant. We are interested in the expectation of system observables $A$ belonging to a $*$-algebra $\mathcal A$ of bounded integral and multiplication operators, 
$$
\langle A\rangle_t = \rho\otimes\omega_\r\big(e^{itH} (A\otimes \bbbone_\r) e^{-it H}\big).
$$
Formally, the `eigenvalues' of $G$ are $G(\vx)$, indexed by $\vx\in\rx^{dN}$. Define the set 
$$
\Gamma :=\big\{ (\vx,\vy)\in\rx^{dN}\!\times\rx^{dN}\ :\ G(\vx)=G(\vy) \big\} \subset \rx^{dN}\times\rx^{dN}
$$
and denote by $\mathbf 1_\Gamma(\vx,\vy)$ its indicator function. We define the projection operator $\mathcal P$ acting on integral operators $A$, with kernels $A(\vx,\vy):\rx^{dN}\!\times\rx^{dN}\rightarrow\cx$, by
$$
[\mathcal P A](\vx,\vy) = \mathbf 1_\Gamma(\vx,\vy) A(\vx,\vy).
$$
For operators $V$ of multiplication by a function $V(\vx)$ we set $\mathcal P V =V$. The formal picture is that $\mathcal P$ keeps invariant the diagonal blocks of an operator $A$ (relative to the eigenbasis of $G$), while setting all off-diagonals to zero. The action of $\mathcal P$ implements the decoherence in the eigenbasis of $G$, that is, the position basis. 
\medskip

We show in Theorem \ref{thm1.0} that $\forall t>0$, $\forall A\in\mathcal A$,
$$
\lim_{\lambda\rightarrow\infty} \langle A\rangle_t = {\rm tr}\big( \rho\, e^{it\mathcal P H_\s}(\mathcal P A)e^{-it\mathcal P H_\s}\big).
$$
The trace is over the system Hilbert space $L^2(\rx^{dN})$. The right hand side is what is sometimes called the Zeno dynamics | the action of a non-selective block-diagonalizing projection $\mathcal P$ followed by the unitary dynamics generated by the block-diagonalized original Hamiltonian $\mathcal P H_\s$. While the Zeno dynamics is commonly viewed as emerging from infinitely frequent quantum measurements on the system \cite{FP08, Fetal2000}, we derive it in Theorem \ref{thm1.0} as the result of the (infinitely) strong coupling with the reservoir. If the set $\Gamma$ has measure zero, in which case we say that $G$ is {\it non-degenerate}, then $\mathcal P H_\s$ and $\mathcal P A$ are multiplication operators ($\mathcal P$ annihilates the non-diagonal operators in $\mathcal A$ in this case) and the dynamics reduces to (Corollary \ref{cor1}) 
$$
\lim_{\lambda\rightarrow\infty}\langle A\rangle_t = {\rm tr}\big(\rho\,\mathcal P A\big),\qquad \forall t>0.
$$
This expresses the Zeno effect as the `freezing' of the system dynamics, or the instantaneous decoherence. The regime  $\lambda\rightarrow\infty$ at fixed time $t>0$, squeezes the decoherence process into an instantaneous effect at $t=0_+$, and it leads to a singularity of the dynamics at the origin: the above right hand side ${\rm tr}(\rho\,\mathcal PA)$ does not generally equal the initial value ${\rm tr}(\rho\, A)$. To resolve the instantaneity of the decoherence process (and the discontinuity) we should balance short times $t$ against large but finite $\lambda$. Let us then consider 
\begin{align*}
t\propto \lambda^{-\alpha} \quad \text{for some $\alpha\ge 0$ and $\lambda\rightarrow\infty$.}
\end{align*}
The case $\alpha=0$ corresponds to the above Zeno regime: $t$ fixed and $\lambda\rightarrow\infty$. We show in Theorem \ref{thm:new2n} that for $0<\alpha<1$ the Zeno effect persists,  $\lim_{\lambda\rightarrow\infty} \langle A\rangle_t = {\rm tr}\big( \rho\, \mathcal P A\big)$, while for $\alpha>1$ the time is too short for the dynamics to have any effect, resulting in the static limit $\lim_{\lambda\rightarrow\infty} \langle A\rangle_t = {\rm tr}\big( \rho\,A\big)$. The critical scaling at which the decoherence process is resolved, is $\alpha=1$. We describe it in Theorem \ref{thm:new2n} (and a refined version in Theorem \ref{thm:new2}). We call 
\begin{align*}
\tau = \lambda t > 0
\quad \text{ with $t\rightarrow 0_+$ and $\lambda\rightarrow\infty$,}
\end{align*}
the {\it fine grained time scale}, or the {\it ultrastrong coupling scaling}. It is the analogue of the coarse grained time scale $\tau'=\lambda^2t$ used in the ultraweak coupling (van Hove) theory, there however $t\rightarrow\infty$ and $\lambda\rightarrow 0$ \cite{VH55,Da74,Da76}. We show in Theorem \ref{thm:new2n} that for any $a\ge 0$, $\lambda>0$  and all integral operators $A\in\mathcal A$,
$$
\sup_{0\le\lambda t\le a} \Big| \langle A\rangle_t
- {\rm tr}\big(\Lambda_{\lambda t}(\rho)\, A\big) \Big|
\le \|A^+\| \,C(a,\lambda).
$$
Here, $\|A^+\|$ is a norm on integral operators and the constant satisfies $\lim_{\lambda\rightarrow\infty}C(a,\lambda)=0$ under generic conditions on $\omega_\r$. The $\Lambda_\tau$, $\tau\ge 0$, is a dynamical map acting on (initial) system density matrices. Expressed on integral kernels,
\begin{align*}
\big[\Lambda_{\tau}(\rho)\big](\vx,\vy) = D_\tau(\vx,\vy)\rho(\vx,\vy),
\end{align*}
where $D_\tau(\vx,\vy)$ is the reservoir {\it decoherence function},
$$
D_\tau(\vx,\vy) = \omega_\r\big(e^{i\tau [G(\vy)-G(\vx)] \varphi(g)}\big).
$$
The result stated above holds for integral operators $A$ having sufficiently regular integral kernels. The integral kernel of a multiplication operator $V$ is formally $V(\vx) \delta(\vx-\vy)$. As $D_\tau(\vx,
\vx)=1$, the dual action of $\Lambda_\tau$ on such operators is the identity. This leads to the following result (Theorem~\ref{thm:new2n}): For all operators $A\in\mathcal A$ of multiplication by a function $V(\vx)$, $a\ge 0,\lambda>0$,
\begin{align*}
\sup_{0\le\lambda t\le a} \big| \langle A\rangle_t
- {\rm tr}\big(\rho V\big) \big|\le \|V\|_\infty\, C'(a,\lambda),
\end{align*}
for an explicit constant satisfying $\lim_{\lambda\rightarrow\infty}C'(a,\lambda)=0$.

For typical reservoir states the decoherence function $D_\tau(\vx,\vy)$ decays in $\tau$ for $(\vx,\vy)$ such that $G(\vx)\neq G(\vy)$ ({\it e.g.}~like a Gaussian function of $\tau$ if $\omega_\r$ is a Gaussian state). This drives the decoherence process for $(\vx,
\vy)\not\in\Gamma$, while $D_\tau(\vx,
\vy)=1$  for $(\vx,\vy)\in\Gamma$ (because $\omega_\r(\bbbone)=1$) reflects that within the `Zeno subspaces' the dynamics is trivial. The map $\Lambda_\tau$ interpolates continuously between the initial state and the Zeno limit (weak limits on $\mathcal A$), 
$$
\lim_{\tau\rightarrow 0}\Lambda_\tau(\rho)=\rho,\qquad \lim_{\tau\rightarrow\infty}\Lambda_\tau(\rho)=  \mathcal P \rho.
$$

We then analyze the markovianity properties of the effective dynamics $\Lambda_t$ in Theorem~\ref{thm_markov}. We show that CP-divisibility and P-divisibility are equivalent for the dynamical map $\Lambda_t$, and they are also equivalent to the ratio $
D_t(\vx,\vy)/D_s(\vx,\vy)$ being a positive definite kernel ($t\ge s\ge 0$). This generalizes previously known results for finite-dimensional, pure dephasing open systems. For Gaussian states $\omega_\r$ the ratio is simply $D_{\sqrt{t^2-s^2}}(\vx,\vy)$ which is a positive definite kernel, implying that $\Lambda_t$ is markovian (both CP- and P-divisible). However, in contrast to the weak coupling theory, or the singular coupling theory\footnote{The {\it singular coupling limit} consists in scaling the form factor so that it becomes increasingly peaked (with divergent $L^2$-norm), resulting in a white noise correlation function for the reservoir. This physically and mathematically different model leads also to a semigroup dynamics \cite{GK,Palmer}. }, the dynamics $\Lambda_t$ is not a semigroup in $t$. 

We use our results on the decoherence to derive the spatial localization of macroscopic quantum objects, that is the fact that large  objects are found to be in spatially localized states (in contrast to microscopic objects, which are typically found in states of definite energy). Consider an object to be made up of $N$ particles with coordinates $x_j\in \rx^d$. Take two wave functions of the object, $\psi_1(\vx)$ and $\psi_2(\vx)$, $\vx=(x_1,\ldots,x_N)$, supported in disjoint sets $\vx\in I_1^N, I_2^N\subset\rx^{dN}$, respectively. The superposition
\begin{align*}
\psi(\vx)=\frac{1}{\sqrt 2}\big(\psi_1(\vx)+\psi_2(\vx)\big)
\end{align*}
is a state in which the object is delocalized over $I_1$ and $I_2$. The object is then coupled to the environment via a potential, such that every particle interacts individually via  $G_1:\rx^d\rightarrow \rx$,
\begin{equation*}
G(\vx)=\sum_{j=1}^N G_1(x_j).
\end{equation*}
The initial object density matrix $\rho=|\psi\rangle\langle\psi|$ evolves into $\Lambda_{\lambda t}(\rho)$ having density matrix kernel (Theorem \ref{thm:new2n}),
\begin{align*}
\big[\Lambda_{\lambda t}(\rho)\big](\vx,\vy) = \frac12 D_{\lambda t}(\vx,\vy)
\big(\psi_1(\vx)+\psi_2(\vx)\big)\big(\overline{\psi_1(\vy)}+\overline{\psi_2(\vy)}\big).
\end{align*}
If the potential $G_1$ resolves the regions $I_1$ and $I_2$ (see Section \ref{sec:spatloc}), then the decoherence function decays in time on a time scale $\tau_\r\propto \frac1N$, uniformly in $\vx$ in one of the $I_{1,2}^N$ and $\vy$ in the other. 
As a consequence, the cross terms $D_{\lambda t}(\vx,\vy)\psi_i(\vx)\overline{\psi_j(\vy)}$, for $i\neq j$ become negligible, 
\begin{align*}
\Lambda_{\lambda t}(\rho) \approx \frac12 \big[\Lambda_{\lambda t}(|\psi_1\rangle\langle\psi_1|)+\Lambda_{\lambda t}(|\psi_2\rangle\langle\psi_2|)\big],\qquad \text{for $\lambda t >\tau_\r$, or $t>t_{\rm loc}:=\tau_\r/\lambda$}.
\end{align*}
The right side is a mixed state describing an ensemble of {\it localized} states. As the localization time satisfies $t_{\rm loc}\propto \frac 1N$, the localization happens much quicker for macroscopic objects ($N$ large). If $G_1$ varies considerably over a typical length $\ell_{\rm loc}$, then quantum superpositions at larger distances are suppressed and objects of the size up to $\ell_{\rm loc}$ become localized. It is known that decoherence caused by the interaction with an environment is at the root of spatial localization of macroscopic objects \cite{Joosetal, Haroche-Raimond, SchlossBook,BassiGhirardi}. Our treatment gives a rigorous analysis of this effect.

\section{Model}

We consider an $N$-body quantum system in $d$ spatial dimensions with Hilbert space
$$
\h_\s = L^2\big(\rx^{dN}\!,d\vx\big)\equiv L^2(\rx^{dN}),\qquad \vx=(x_1,\ldots,x_N)\in\rx^{dN}.
$$
A density matrix | or state | on $\h_\s$ is a non-negative operator  $\rho$ of unit trace: $\rho\ge 0$, ${\rm tr}\rho=1$. If $\rho$ has rank one, $\rho=|\psi\rangle\langle\psi|$ for some normalized $\psi\in\h_\s$, then it is called a pure state, otherwise $\rho$ is a mixed state. The (Schr\"odinger) dynamics of a  state $\rho$ is given by
\begin{align}
\label{nHs}
\rho_t =e^{-itH_\s}\rho\, e^{itH_\s},
\end{align}
where the generator $H_\s$, the Hamiltonian, is a self-adjoint operator on $\h_\s$. The second component of the quantum complex we consider is a {\it reservoir} (environment) modeled by a bosonic quantum field. It is described by its algebra of observables, the Weyl CCR $C^*$-algebra $\mathcal W_\r$ over the single-particle space of test functions $L^2(\rx^3, d^3k)$, that is, the unital $C^*$-algebra generated by Weyl operators $W(f)$ satisfying the CCR
\begin{equation}
\label{bog1}
W(f)W(g) = e^{-\frac i2{\rm Im}\langle f,g\rangle} W(f+g),\qquad f,g\in\hh.
\end{equation}
The dynamics of the reservoir is determined by the Bogoliubov transformation $f\mapsto e^{i\omega t}f$, or
\begin{equation}
\label{weyldyn1}
W(f)\mapsto  W\big(e^{i\omega t }f\big),
\end{equation}
where 
\begin{equation}
\label{disp}
\omega=\omega(k)\ge 0
\end{equation}
is a function of $k\in\rx^3$, called the dispersion relation. A state of the reservoir is a positive linear functional $\omega_\r$ on $\mathcal W_\r$ normalized as $\omega_\r(\bbbone)=1$, where $\bbbone=W(0)$ is the unit of $\mathcal W_\r$. Positivity means that $\omega_\r(X^*X)\ge 0$ for any $X\in\mathcal W_\r$. Unlike the system, we describe the reservoir by its observable algebra. We do this so we can cover inequivalent reservoir states $\omega_\r$ on the same footing (such as equilibrium states at different temperatures, coherent states, Fock states).

The interacting system-reservoir complex is defined by specifying an interacting dynamics. From the physical perspective this is done by adding to the `uncoupled Hamiltonian' an interaction term involving creation and annihilation operators on the reservoir side. They drive exchange processes  of energy and other quantities between $\s$ and $\r$. However, in the algebraic setting, the reservoir does not have a Hamiltonian and creation and annihilation operators are not defined as objects in $\mathcal W_\r$ | but they are well defined as operators in the representation Hilbert space of sufficiently nice states.  Let $\omega$ be a state on $\mathcal W_\r$ and denote its GNS triple by $(\h_\omega,\pi_\omega,\Omega_\omega)$, so that $\omega(W(f))=\langle \Omega_\omega, \pi_\omega(W(f))\Omega_\omega\rangle_{\h_\omega}$. The state  is called {\it regular} if $\lim_{t\rightarrow 0}\omega(W(tf))=1$ for all $f\in\hh$. This is equivalent to saying that the map $\rx\ni t\mapsto \pi_\omega(W(tf))$ is strongly continuous on $\h_\omega$. This map then defines a strongly continuous unitary group  generated by a self-adjoint operator $\varphi(f)$ on $\h_\omega$, that is, $\pi_\omega(W(tf)) =e^{it\varphi(f)}$. The generator $\varphi(f)$ is the field operator and can be written as the sum of a creation and an annihilation operator in the usual way. The commutation relation  $[\varphi(f),\varphi(g)]=i {\rm Im}\langle f,g\rangle_{\hh}$ is inherited from \eqref{bog1}. 

We consider initial $\s\r$ states on $\mathcal B(\h_\s)\otimes \mathcal W_\r$, of the form 
\begin{align}
\label{srin}
\omega_0=\rho\otimes \omega_\r,
\end{align}
where $\rho$ is the initial density matrix of $\s$ acting on $\h_\s$ and we identify $\rho$ with the map ${\rm tr}(\rho\,\cdot)$ on $\mathcal B(\h_\s)$. Throughout we make the following condition on the reservoir state $\omega_\r$.
\begin{itemize}
\item[\bf(A0)] The state $\omega_\r$ is {\it regular} and the dynamics $t\mapsto W(e^{i\omega t}f)$ is {\it implementable}, which means that there is a self-adjoint operator $H_\r$ on the GNS space $\h_\r$ of $\omega_\r$ such that 
\begin{equation}
\label{impl}
\pi_\r\big(W(e^{i\omega t}f)\big) = e^{i t H_\r} \pi_\r\big(W(f)\big) e^{-itH_\r}.
\end{equation}
Here, $(\h_\r,\pi_\r,\Omega_\r)$ denotes the GNS triple associated to $\omega_\r$. 
\end{itemize}
Due to the regularity of $\omega_\r$ there are self-adjoint field operators $\varphi(f)$ on $\h_\r$, determined by
\begin{align}
\label{fieldop}
\pi_\r(W(f)) = e^{i\varphi(f)},\quad f\in\hh.
\end{align}
Now we are in a position to consider the joint $\s\r$  Hilbert space 
$$
\h_{\s\r} = \h_\s\otimes\h_\r
$$
and to define the interacting Hamiltonian 
\begin{equation}
\label{c1}
H = H_\s + H_\r +\lambda G\otimes\varphi(g).
\end{equation}
We omit trivial tensor factors in the expressions. The coupling constant $\lambda$ in \eqref{c1} is a real number, $\varphi(g)$ is the field operator \eqref{fieldop} smoothed out with a `form factor'
\begin{equation}
\label{ff}
g=g(k)\in L^2(\rx^3,d^3k).
\end{equation}
The system interaction operator $G$ in \eqref{c1} is a potential, that is, the operator of multiplication by a real, measurable function $G(\vx)$,
$$
G\psi(\vx) = G(\vx)\psi(\vx), \quad \psi\in L^2(\rx^{dN},d\vx).
$$
The Hamiltonian $H$, \eqref{c1} generates a Heisenberg dynamics
\begin{align}
O\mapsto O_t = e^{itH}Oe^{-itH},\qquad O\in\mathcal B(\h_{\s\r})   
\end{align}
and we would like to make sense of the average at time $t$, of system observables 
$$
O=A\otimes \bbbone_\r.
$$
Formally the average is given by 
\begin{align} 
\label{n9}
\langle A\rangle_t \ \   =\  \text{``}\ \rho\otimes\omega_\r(e^{itH}(A\otimes\bbbone_\r)e^{-itH}).\ \text{''}
\end{align}
The definition of the quantity on the right side of \eqref{n9} needs some care because $e^{itH}(A\otimes\bbbone_\r)e^{-itH}$ generally is not an element of $\mathcal B(\h_\s)\otimes\mathcal W_\r$ but rather it belongs to the weak closure of $\mathcal B(\h_\s)\otimes\mathcal \pi_\r(\mathcal W_\r)$ in $\mathcal B(\h_{\s\r})$. We present in Section \ref{sec:defdyn} a natural way to define \eqref{n9} by using a Dyson series expansion, 
\begin{align}
\label{defdyn'}
\langle A\rangle_t 
\equiv \sum_{n\ge 0} i^n \int_{0\le t_n\le\cdots\le t_1\le t} \rho\otimes \omega_\r\big( B_{t,t_1,\ldots,t_n;A}\big),
\end{align}
where $B_{t,t_1,\ldots,t_n;A}$ can be written in terms of operators in $\mathcal B(\h_\s)\otimes \mathcal W_\r$ on which the state $\rho\otimes \omega_\r$ has a meaning.
\medskip

Next we introduce the class of system operators we consider. An integral operator $T$ on $\h_\s=L^2(\rx^{dN})$ is given by the expression $T\psi(\vx) = \int T(\vx,\vy)\psi(\vy)d\vy$, where $T(\vx,\vy): \rx^{dN}\times\rx^{dN}\rightarrow\cx$ is a measurable function, called the integral kernel. Let $T$ be an integral operator on $L^2(\rx^{dN})$ with kernel $T(\vx,\vy)$. We denote by $T^+$ the integral operator having the kernel $|T(\vx,\vy)|$,
\begin{align}
\label{T+}
T^+\psi(\vx) = \int_{\rx^{dN}} |T(\vx,\vy)| \psi(\vy) d\vx.
\end{align}
We define the following sets of operators:
\begin{align*}
\mathcal I_+ & \quad \mbox{is the set of integral operators $T$ on $\h_\s$ s.t. $\|T^+\|<\infty$}\nonumber\\
\mathcal V \ & \quad \mbox{is the set of multiplication operators by functions $V(\vx):\rx^{dN}\rightarrow \cx$ s.t. $\|V\|_\infty<\infty$}
\end{align*}
Here, $\|\cdot\|$ denotes the operator norm of bounded operators on $\h_\s$ and $\|V\|_\infty$ is the essential supremum of the function $V(\vx)$.  Since  $\|T\|\le \|T^+\|$ the operators in $\mathcal I_+$ are bounded. The multiplication operator $A\in\mathcal V$ associated to the function $V(\vx)$ is also bounded and has operator norm $\|A\|=\|V\|_\infty$. A criterion for the boundedness of $T^+$ is shown in \cite{HalmosSunder} (Theorem 10.5): 
\begin{equation}
\label{tr4}
\|T^+\|<\infty \quad \Longleftrightarrow\quad  \int_{\rx^{dN}\!\times\rx^{dN}} \big|\psi(\vx) T(\vx,\vy)\phi(\vy)\big| d\vx d\vy<\infty \quad \forall \psi,\phi\in\h_\s.
\end{equation}
We show the following result in Section \ref{sec:proofs}.
\begin{prop}
\label{prop:intop}
Let $\mathcal A$ denote the collection of sums of products of elements from $\mathcal I_+$ and $\mathcal V$. Then $\mathcal A$ is a $*$-algebra of bounded operators on $\h_\s$. Also, products of elements of $\mathcal I_+$ and $\mathcal V$ belong either to $\mathcal I_+$ or to $\mathcal V$. 
\end{prop}

{\bf Assumptions.} Recall that the system dynamics is given by \eqref{nHs}. We assume that 
\begin{itemize}
\item[\bf(A1)] The system Hamiltonian satisfies $H_\s\in\mathcal A$.

\item[\bf(A2)] The initial system density matrix \eqref{srin} is of the form
\begin{align}
\label{indmat}
\rho = \sum_{j=1}^J p_j|\psi_j\rangle\langle\psi_j|
\end{align}
for some finite $J$, where $0\le p_j\le 1$, $\sum_{j=1}^Jp_j=1$ and where the $\psi_j\in L^2(\rx^{dN})\cap L^1(\rx^{dN})$ are a family of functions normalized as $\|\psi_j\|_{L^2}=1$. The family is not required to be orthogonal in $L^2$.
\end{itemize}
The rank of the density matrix \eqref{indmat} is at most $J$. The integral kernel of $\rho$ is 
$$
\rho(\vx,\vy) = \sum_{j=1}^J p_j \psi_j(\vx)\overline{\psi_j(\vy)}
$$
and belongs to $L^1(\rx^{dN}\!\times \rx^{dN})\cap L^2(\rx^{dN}\!\times\rx^{dN})$.  Our third assumption concerns the dispersion relation $\omega(k)$ of the reservoir, \eqref{disp} and the form factor $g(k)$, \eqref{c1}. 
\begin{itemize}
\item[\bf (A3)] The expectation 
$\omega_\r\big(W(\frac{e^{i\omega t}-1}{\omega}g)\big)$ is well defined for all  $t\ge 0$.
\end{itemize}

\begin{example}[Three classes of canonical reservoir states] 
\label{ex2.1}
\begin{itemize}
\item[-]  A (centered) {\bf Gaussian state} is characterized by a self-adjoint (generally unbounded) covariance operator $\mathcal C\ge \bbbone$ on $L^2(\rx^d,d^3k)$,
\begin{equation}
\label{c17}
\omega_\r(W(f)) = e^{-\frac14\langle f,\mathcal Cf\rangle},\qquad f\in{\rm dom}(\sqrt\mathcal C).
\end{equation}
Gaussian states are regular states. Equilibrium (KMS) states at all temperatures $T\ge 0$ are Gaussian.  The GNS representation of Gaussian states is explicit (Araki-Woods representation \cite{BR}). The condition (A3) is
$\mathcal C^{1/2} \  \frac{e^{i\omega(k)t}-1}{\omega(k)}g(k)\in L^2(\rx^3,d^3k)$.

\item[-] A {\bf coherent state} is given by $\omega_\r(\cdot) = \langle W(\alpha)\Omega, (\cdot) W(\alpha)\Omega\rangle$ for some $\alpha\in L^2(\rx^3,d^3k)$, where  $\Omega$ the vacuum vector in the symmetric Fock space $\h_\r=\bigoplus_{N\ge 0} \mathcal S{\mathfrak h}^{\otimes N}$ over the one-particle space $\mathfrak h=L^2(\rx^3,d^3k)$ (the Fock space is the GNS Hilbert space). Its generating functional is
\begin{align}
\omega_\r(W(f)) = e^{-\frac14 \|f\|^2_{L^2}} e^{i{\rm Im}\langle \alpha,f\rangle}.
\end{align}
Coherent states are regular. The condition (A3) is
$\frac{e^{i\omega(k)t}-1}{\omega(k)}g(k)\in L^2(\rx^3,d^3k)$.

\item[-] A {\bf Fock state} is given by $\omega_\r(\cdot)= \langle \Psi_N, (\cdot)\Psi_N\rangle$, where $\Psi_N$ is a normalized vector in Fock space  having $N$ particles in the state $h\in L^2(\rx^3,d^3k)$, $\psi_N=\frac{1}{\sqrt N!}[a(h)^*]^N\Omega$ (and $\Omega$ the vacuum vector). We have
\begin{align}
\omega_\r(W(f))=
e^{-\frac14\|f\|_{L^2}^2}
L_N(\tfrac{|\langle h,f\rangle|^2}{2}),
\end{align}
where $L_N$ is the Laguerre polynomial of order $N$ ({\it cf.}~\eqref{laguerre}). Also in this case the condition (A3) reads
$\frac{e^{i\omega(k)t}-1}{\omega(k)}g(k)\in L^2(\rx^3,d^3k)$.
\end{itemize}
\end{example}

\begin{example}[$H_\s$ satisfying (A1)] \!\!\!Consider a system Hamiltonian $H_\s$ having a kinetic and a potential term,
$$
H_\s(\mathbf p,\vx) =  T(\mathbf p) + U(\vx),
$$
where $T, U:\rx^{dN}\rightarrow\rx$ and $\mathbf p=(p_1,\ldots,p_N)$, $p_j=-i\nabla_{\!x_j}$. This Hamiltonian satisfies Condition (A1) provided $\widecheck{T}\in L^1(\rx^{dN})$ (inverse Fourier transform) and $U\in L^\infty(\rx^{dN})$. The sufficiency of the boundedness condition on $U$ is clear.  The action of $T$ is defined using the Fourier ($\ \widehat{}\ $) and inverse Fourier ($\ \widecheck{}\ $) transforms,  $T(\mathbf p)\psi(\vx)= \big(T(\vk)\widehat \psi(\vk)\big)\widecheck{\  }(\vx)$. One arrives readily at the expression for the integral kernel of $T(\mathbf p)$,
$[T(\mathbf p)](\vx,\vy) = (2\pi)^{-dN/2} \,\widecheck{T}(\vx-\vy)$. 
Using the Cauchy-Schwarz inequality, we have for all $\psi,\phi\in\h_\s$,
\begin{eqnarray*}
\lefteqn{\int |\psi(\vx)| |\widecheck{T}(\vx-\vy)| |\phi(\vy)| d\vx d\vy}\nonumber\\
&\le& \Big(\int |\psi(\vx)|^2 |\widecheck{T}(\vx-\vy)|d\vx d\vy\Big)^{1/2} 
\Big(\int |\phi(\vy)|^2 |\widecheck{T}(\vx-\vy)|d\vx d\vy\Big)^{1/2}=\|\widecheck{T}\|_{L^1} \|\psi\|\, \|\phi\|.
\end{eqnarray*}
By \eqref{tr4} this means that $\|T^+\|<\infty$ for  $\widecheck{T}\in L^1(\rx^{dN})$.
\end{example}

\section{Results}

\subsection{Instantaneous spatial decoherence}

We investigate averages $\langle A\rangle_t$, \eqref{defdyn'} in the regime where time $t>0$ is fixed and the coupling constant $\lambda\rightarrow\infty$. This is called the {\it Zeno limit} (sometimes also the ultra-strong coupling limit) \cite{Trush+,Marcantoni-Merkli}. We introduce the projection $\dg$ acting on observables in $\mathcal A$ which eliminates coherences in the `eigenbasis of $G$', defined in the following way. Given an integral operator $T$ on $\h_\s$ with integral kernel $T(\vx,\vy):\rx^{dN}\times\rx^{dN}\rightarrow\cx$, the operator $\dg T$ is the integral operator having the kernel 
\begin{equation}
\label{i1}
[\dg T](\vx,\vy) = \ind_\Gamma(\vx,\vy) T(\vx,\vy),
\end{equation}
where
\begin{equation}
\label{i2}
\ind_\Gamma(\vx,\vy) = \left\{
\begin{array}{ll}
1 & \mbox{if \, $(\vx,\vy)\in\Gamma $}\\
0 & \mbox{if \, $(\vx,\vy)\not\in \Gamma$}
\end{array}
\right. 
\end{equation}
is the indicator function of the set
\begin{equation}
\label{Gamma}
\Gamma :=\big\{ (\vx,\vy)\in\rx^{dN}\!\times\rx^{dN}\ :\ G(\vx)=G(\vy) \big\} \subseteq \rx^{dN}\times\rx^{dN}.
\end{equation}
\begin{lem}
\label{lem:3}
The set $\Gamma$ is measurable and $\dg$ 
only depends on the equivalence class of $G$.
More precisely, let $G_1, G_2:\rx^{dN}\rightarrow\rx$ and denote by $\dg_j,\mathbf 1_{\Gamma_j}, \Gamma_j$, $j=1,2$, the associated quantities \eqref{i1}-\eqref{Gamma}. If $G_1=G_2 $ a.e.~(Lebesgue $\rx^{dN}$\!), then $\mathbf 1_{\Gamma_1}=\mathbf 1_{\Gamma_2}$ a.e.~(Lebesgue $\rx^{dN}\!\times\rx^{dN}$) and $\dg_1=\dg_2$.
\end{lem}

{\bf Proof of Lemma \ref{lem:3}.} We have $\Gamma=H^{-1}(\{0\})$, where $H(\vx,\vy)\equiv G(\vx)-G(\vy)$. Since $H$ is a measurable function, $\Gamma$ is a measurable set. Next we show that the symmetric difference $\Gamma_1\Delta\Gamma_2=(\Gamma_1\backslash \Gamma_2)\cup (\Gamma_2\backslash \Gamma_1)$ has measure zero ($\rx^{dN}\times \rx^{dN}$). Set $H_j(\vx,\vy)=G_j(\vx)-G_j(\vy)$, $j=1,2$. Then we have $H_1=H_2$ a.e.~($\rx^{dN}\times\rx^{dN}$). Indeed, Let $N=\{(\vx,\vy) : H_1(\vx,\vy)\neq H_2(\vx,\vy)\}$ and let $M=\{\vx : G_1(\vx)\neq G_2(\vx)\}$. If $(\vx,\vy)\in N$ then either $G_1(\vx)\neq G_2(\vx)$ or $G_1(\vy)\neq G_2(\vy)$, so $N\subseteq (M\times \rx) \cup (\rx\times M)$. Now $M\times \rx=\cup_{n\in\mathbb N}\, M\times (-n,n)$ is a countable union of sets of measure zero ($\rx^{dN}\times\rx^{dN}$) so $M\times\rx$ has measure zero. $\rx\times M$ too has measure zero. Thus $N$ has measure zero, so $H_1=H_2$ a.e.~($\rx^{dN}\times\rx^{dN}$).  Next we show the symmetric difference has measure zero. Let $(\vx,\vy)\in\Gamma_1\backslash\Gamma_2$. Then $H_1(\vx,\vy)=0$ and $H_2(\vx,\vy)\neq 0$, so $(\vx,\vy)\in N$. Thus $\Gamma_1\backslash\Gamma_2\subseteq N$. In the same way, $\Gamma_2\backslash\Gamma_1\subseteq N$. So $\Gamma_1\Delta\Gamma_2\subseteq N$ has measure zero. Finally, it is easy to see that $\Gamma_1\Delta\Gamma_2 =\{(\vx,\vy) : \mathbf 1_{\Gamma_1}(\vx,\vy)\neq \mathbf 1_{\Gamma_2}(\vx,\vy)\}$. So $\mathbf 1_{\Gamma_1}=\mathbf 1_{\Gamma_2}$ a.e.~($\rx^{dN}\times\rx^{dN}$) indeed. Modifying the integral kernel on a set of measure zero does not alter the associated integral operator, so $\dg_1=\dg_2$. \hfill\qed
\medskip

The projection $\dg$ leaves $\mathcal I_+$ invariant. We define the action of $\dg$ on $\mathcal V$ as $\dg V=V$ for $V\in\mathcal V$, in particular, $\dg \bbbone=\bbbone$. We extend the action of $\dg$ to $\mathcal A$ by linearity (see Proposition \ref{prop:intop}). $\dg$ leaves $\mathcal A$ invariant. For the next result we require an effective coupling condition.
\begin{itemize}
\item[{\bf (A4)}]
The reservoir dispersion relation \eqref{disp} satisfies $\omega(k_0)>0$ for some $k_0\in\rx^3$ and furthermore, the form factor $g(k)$, \eqref{ff} is continuous and not constant as a function of the radial variable $|k|$ in some open interval containing $|k_0|$.
\end{itemize}
The condition (A4) ensures that the effect of taking $\lambda$ large is not offset by an accidental decoupling of the system and reservoir. This condition is also present in the previous related work \cite{Marcantoni-Merkli}.
\begin{thm}[\bf Spatial decoherence]
\label{thm1.0}
Let $\omega_\r$ be a Gaussian state \eqref{c17} and assume that the conditions (A1)-(A4) hold. 
Then we have for every $t>0$ and every $A\in\mathcal A$,
\begin{equation}
\label{31-1}
\lim_{\lambda\rightarrow\infty}\langle A\rangle_t = {\rm tr}\big(  \rho \, e^{it\dg H_\s} ( \dg A ) e^{-it \dg H_\s}\big).
\end{equation}
\end{thm}
Informally speaking, $G(\vx)$ are the `eigenvalues' of the multiplication operator $G$ and the set $\Gamma$, \eqref{Gamma} encodes the degeneracy of these values. If $\Gamma$ has measure zero, $|\Gamma|=0$, then we call $G$ {\it non-degenerate}. In this case we have $\dg T=0$ (zero operator) for any $T\in\mathcal I_+$. However, by definition $\dg V=V$ for any $V\in\mathcal V$, even if $|\Gamma|=0$. In the non-degenerate case the action of $\dg$ on $A\in\mathcal A$ `filters out' the diagonal (multiplication operator) part in $A$, while the off-diagonal part ($\mathcal I_+$) is annihilated by the action of $\dg$. So the spatial coherences are eliminated.

For non-degenerate $G$, $|\Gamma|=0$, any operator $\dg A$ for $A\in\mathcal A$ is a multiplication operator, thus $e^{it \dg H_\s}$ and $\dg A$ commute. Then we obtain from \eqref{31-1} the following result.

\begin{cor}[\bf Frozen dynamics | Zeno effect]
\label{cor1}
Assume the conditions of Theorem \ref{thm1.0} and additionally, that $G$ is non-degenerate. Then for every $t>0$ and every $A\in\mathcal A$ we have the time-independent limit
\begin{equation}
\label{corresult}
\lim_{\lambda\rightarrow\infty} \langle A\rangle_t = {\rm tr}\big(\rho \dg A\big).
\end{equation}
\end{cor}
Corollary \ref{cor1} shows that the effect of the coupling in the Zeno regime ($t>0$ fixed, $\lambda\rightarrow\infty$) is to act with $\dg$ on observables, with no further change in time. The action of $\dg$ can be understood as a continuous-variable non-selective measurement of the observable $G$. To elucidate the meaning of this we mention how the result looks like if $\h_\s$ is finite-dimensional (see \cite{Marcantoni-Merkli}). In that case $G$ has the spectral representation $G=\sum_{j} \gamma_j P_j$, where the  $\gamma_j$ are the distinct eigenvalues and the $P_j$ are the eigenprojections. The action of $\dg$ on an operator $A$ is to block-diagonalize $A$ in the decomposition $\h_\s=\oplus_{j} {\rm Ran} P_j$, namely $\dg A=\sum_j P_j A P_j$. The dynamics on the right side of \eqref{31-1} happens independently (block-diagonally) within the Zeno subspaces ${\rm Ran} P_j$. The coherences between different subspaces are eliminated by the action of $\dg$. If the finite-dimensional $G$ has non-degenerate spectrum, $\dim{\rm Ran} P_j=1$ for all $j$, then the blocks are of size $1\times 1$ and the dynamics becomes time-independent, just as in \eqref{corresult}.
\smallskip

\begin{example} Consider a step potential $G(\vx)=\mathbf 1_\Omega(\vx)$ for some measurable $\Omega\subset\mathbb R^{dN}$, so $G(\vx)=1$ for $\vx\in \Omega$ and otherwise $G(\vx)=0$. Then $\Gamma = \Omega\times \Omega \cup \Omega^c\times \Omega^c$ (complement) and so $G$ is {\it not} non-degenerate. As $\Gamma$ is the disjoint union of two sets the projection is the sum of two commuting projections, $\mathcal P=\mathcal P_\Omega+\mathcal P_{\Omega^c}$, were $\mathcal P_\Omega$ acts on integral kernels by multiplication with $\mathbf 1_{\Omega\times \Omega}(\vx,\vy)$, and similar for $\mathcal P_{\Omega^c}$. There are two blocks evolving independently for the Zeno dynamics. As another example, let $\vx=x\in\rx$ and suppose $G(x)$ is strictly monotonic. Then $\Gamma=\cup_{x\in\rx}(x,x)$ so that $G$ is non-degenerate. The Zeno dynamics is trivial, given by \eqref{corresult}.
\end{example}

\subsection{Temporal resolution of the decoherence process}

Theorem \ref{thm1.0} shows that for $t>0$ fixed and $\lambda\rightarrow\infty$, the dynamics of observables is given by an projective measurement, or decoherence, $\dg A$ and a subsequent Hamiltonian dynamics generated by the (block-)diagonalized Hamiltonian $\dg H_\s$. The collapse of the coherence is instantaneous in this regime and it leads to a discontinuity in the evolution of non-diagonal observables,  $\lim_{t\rightarrow 0_+}\lim_{\lambda\rightarrow\infty}\langle A\rangle_t \neq \langle A\rangle_0$ for $A\in \mathcal I_+$. In order to follow the very fast process of decoherence for large $\lambda$, we have to look at a finer time-scale of short times. We call this the resolution of the decoherence process. To describe this process we introduce the {\it decoherence function}
\begin{equation}
\label{decof}
D_t(\vx,\vy) = \omega_\r\Big(W\big(t[G(\vy)-G(\vx)]g\big) \Big),\qquad t\ge 0, \ \vx,\vy\in\rx^{dN}
\end{equation}
with $W$ the Weyl operator \eqref{fieldop}. Given an integral kernel $S\in L^2(\rx^{dN}\!\times \rx^{dN})$ we denote by $O_S$ the corresponding integral operator. On such operators we define the linear map $\Lambda_t$ by
\begin{equation}
\label{defL}
\Lambda_t(O_S) = O_{\mathcal D_t S},\qquad
\mathcal D_t: S(\vx,\vy)\mapsto  D_t(\vx,\vy) S(\vx,\vy).
\end{equation}
Recall that a density matrix $\rho$ on $L^2(\rx^{dN})$ is a non-negative operator $\rho\ge0$ of unit trace, ${\rm tr}(\rho)=1$. Any density matrix $\rho$, being trace-class, is automatically Hilbert-Schmidt and so it is given by an integral kernel $\rho(\vx,\vy)\in L^2(\rx^{dN}\!\times \rx^{dN})$. As $D_t(\vx,\vx)=1$ we have $\int_{\rx^{dN}}S(\vx,\vx) D_t(\vx,\vx) d\vx=\int_{\rx^{dN}}S(\vx,\vx)d\vx$;  if both integrands are continuous functions then the integrals represent traces of the associated integral operators $O_{\mathcal D_t S}$ and $O_S$ and so $\Lambda_t$ is trace preserving on such operators.

\begin{lem}
\label{lem:0}
Fix $t\ge 0$ and suppose that $D_t(\vx,\vy)$ is continuous in $\vx,\vy$. Then the map $\Lambda_t$ leaves invariant the set of density matrices having continuous integral kernels.
\end{lem}
We show more general invariance properties of $\Lambda_t$ in Proposition \ref{prop:two} below. To state our next result we assume a natural decay bound on the state $\omega_\r$ involving the function
\begin{align}
g_t\equiv g_t(k)=\Big(\frac{e^{i\omega (k)t}-1}{i\omega(k) t}-1\Big)\frac 1t g(k).
\end{align}
We have 
$\big(\frac{e^{i\omega t}-1}{i\omega t} - 1\big)\frac1t=-\frac{1}{i\omega t^2}\int_0^{\omega t}dx \int_0^xds e^{-is}$ and $\big|\big(\frac{e^{i\omega t}-1}{i\omega t} - 1\big)\frac1t\big|\le \frac12\omega$, so $\|g_t\|_{L^2}\le \frac12 \|\omega g\|_{L^2}$, which is assumed to be finite. Our next assumption is,

\begin{itemize}
\item[{\bf (A5)}]  There is a function $\mu_\r: [0,\infty)\times [0,\infty) \rightarrow [0,\infty)$ such that  
\begin{align}
\label{mudef}
\Big| 1-\omega_\r\big( W\big(\xi g_t \big)\big)\Big|^{1/2}\le \frac{1}{\sqrt 2} \mu_\r(|\xi|,t),\qquad \xi\in\rx,\ t\ge 0.
\end{align}
Moreover for all $t\ge 0$, $x\mapsto \mu_\r(x,t)$ is monotone increasing and $\mu_\r(0,t)=0$.
\end{itemize}
For $a\ge 0$, $\lambda>0$ set
\begin{align}
\label{Cal}
C(a,\lambda) = \sup_{0\le \tau\le a}\mu_\r\Big(2\|G\|_\infty \frac{a^2}{\lambda} ,\frac{\tau}{\lambda}\Big)+ \frac43 \frac{a^3}{\lambda} \|G\|_\infty^2 \|\sqrt \omega g\|^2_{L^2}+2\frac{a}{\lambda}\|H^+_\s\|e^{2a\|H_\s^+\|/\lambda}\ .
\end{align}

\begin{thm}[\bf Temporal resolution of the decoherence process]
\label{thm:new2n}
Assume (A0)-(A3).
\begin{itemize}
\item[\rm 1.] 
For any $A\in\mathcal V$ a multiplication operator by a function $V(\vx)$ we have
\begin{align}
\label{m29.1n}
 \big| \langle A\rangle_t
- {\rm tr}\big(\rho A\big) \big|\le 2t  \,e^{2t\|H_\s^+\|} \|V\|_\infty.
\end{align}
\item[2.]  Assume in addition (A5) and that $\rho(\vx,\vy), D_t(\vx,\vy)$ are continuous in $(\vx,\vy)$, for all $t\ge 0$. Let $A\in\mathcal I_+$ be such that $A(\vx,\vy)$ is bounded and continuous in $(\vx,\vy)$. 
Set $t=\tau \lambda^{-\alpha}$ for $\alpha>0$ and $\tau\ge 0$. 
\begin{itemize}
\setlength{\itemsep}{0.5em}
\item[\rm (i)] Let $0<\alpha<1$. Then   $\forall \tau>0$: $\lim_{\lambda\rightarrow\infty} \big|\langle A\rangle_t-{\rm tr}(\rho\mathcal P A)\big|=0$ (Zeno effect). 

\item[\rm (ii)] Let $1<\alpha$. Then $\forall \tau\ge0$: $\lim_{\lambda\rightarrow\infty} \big|\langle A\rangle_t-{\rm tr}(\rho A)\big|=0$ (trivial dynamics).

\item[\rm (iii)] Let $\alpha=1$. Then for any $a\ge 0$, $\lambda>0$,
\begin{align}
\label{1m1}
 \sup_{0\le \lambda t\le a}&\big| \langle A\rangle_t
- {\rm tr}\big(\Lambda_{\lambda t}(\rho) A\big) \big|\le\|A^+\| \,C(a,\lambda).
\end{align}
\end{itemize}
\end{itemize}
\end{thm}

The operator $\Lambda_{\lambda t}(\rho)$ in \eqref{1m1} is a density matrix due to Lemma \ref{lem:0}. Each observable $A\in\mathcal A$ is a sum of operators in $\mathcal V$ and in $\mathcal I_+$ (see Proposition \ref{prop:intop}) so by linearity, Theorem \ref{thm:new2n} describes the dynamics $\langle A\rangle_t$ for all $A\in\mathcal A$. 
The continuity of $D_t(\vx,\vy)$ required for ${\it 2.}$~holds true in physically relevant settings as shown in the following example.

\begin{example} If $G(\vx)$ is continuous and $\omega_\r$ is a regular state, then $D_t(\vx,\vy)$ is continuous in $\vx,\vy$ for all $t$. In particular, for the classes of states given in Example \ref{ex2.1} we have,
\begin{equation}
\label{gausdecfun}
D_t(\vx,\vy) = \left\{
\begin{array}{ll}
e^{-\frac14 t^2 (G(\vx)-G(\vy))^2 \, \| \sqrt \mathcal C g\|^2_{L^2}} & \text{Gaussian}\\ [1.5ex]
e^{-\frac14 t^2 (G(\vx)-G(\vy))^2 \, \|g\|^2_{L^2}}\ e^{it (G(\vx)-G(\vy)){\rm Im}\langle \alpha,g\rangle}& \text{coherent}\\[1.5ex]
e^{-\frac14 t^2 (G(\vx)-G(\vy))^2 \, \|g\|^2_{L^2}}\ L_N\big(\tfrac12 t^2(G(\vx)-G(\vy))^2|\langle h,g\rangle|^2\big)& \text{Fock}\\
\end{array}
\right.
\end{equation}
\end{example}

The bound \eqref{1m1} is useful if $\lim_{\lambda\rightarrow \infty}C(a,\lambda)=0$. This is the case for all the three examples above, as detailed in the following result. 

\begin{prop}
\label{prop2n}
We have the following explicit expressions for Gaussian, coherent and Fock states ({\rm cf.}~Example \ref{ex2.1}).
\begin{itemize}
\item[-] For a {\bf Gaussian state} with covariance $\mathcal C$ we can take
\begin{align}
\label{i10n}
\mu_\r(x,t)=\frac{1}{\sqrt 2}\, x\, \|\mathcal C^{1/2} g_t\|_{L^2}.
\end{align}
If $\mathcal C$ is diagonal, that is, an operator of multiplication by a function of $k$, then we can take  $\mu_\r(x,t)=\frac{1}{2\sqrt 2}\, x\, \|\mathcal C^{1/2}\omega g\|_{L^2}$.

\item[-] For a {\bf coherent state}  $\omega_\r(\cdot) = \langle W(\alpha)\Omega, (\cdot) W(\alpha)\Omega\rangle$ we can take
\begin{equation}
\label{ratecohern}
\mu_\r(x,t)=\frac{1}{2\sqrt 2}\, x \,\|\omega g\|_{L^2} +\sqrt x \, \sqrt{\|\alpha\|_{L^2} \|\omega g\|_{L^2}} .
\end{equation}

\item[-] For a {\bf Fock state} $\omega_\r= \langle \Psi_N, (\cdot)\Psi_N\rangle$, $\psi_N=\frac{1}{\sqrt N!}[a(h)^*]^N\Omega$ we can take
\begin{align}
\label{ratefockn}
\mu_\r(x,t) &=\frac{1}{2\sqrt 2} x\, \|\omega g\|_{L^2}\Big[ 1 + 2^{N+1} \|h \|^2_{L^2} \Big( 1 + \Big(\frac{x^2}{8} \|\omega g\|^2_{L^2} \|h \|^2_{L^2} \Big)^{\!N-1}\Big) \Big]^{1/2}.
 \end{align}
\end{itemize}
\end{prop}
We give a proof of Proposition \ref{prop2n} in Section \ref{sec:ref}. 
For all three classes of states in Proposition~\ref{prop2n}, $\mu_\r(x,t)$ is independent of $t$ so that the supremum in $C(a,\lambda)$, \eqref{Cal}, is superfluous. The relations \eqref{i10n}, \eqref{ratecohern} and \eqref{ratefockn}, together with \eqref{Cal}, provide an explicit rate of convergence in $\lambda$ for the effective dynamics \eqref{1m1}. We are interested in the behaviour of $\mu_\r$ for small values of $x= 2 \| G\|_\infty a^2/\lambda$ ({\it cf.} \eqref{Cal}). For the Gaussian state $\mu_\r$ is proportional to $x$, and so is the leading term in \eqref{ratefockn} for the Fock state, while for the coherent state the leading term is  $\sqrt{x}$. Together with the other bounds in \eqref{Cal}, both decaying as $1/\lambda$, we have established the following asymptotics for large coupling
\begin{align*}
C(a,\lambda)  \sim\left\{
\begin{array}{cc}
\frac{1}{\lambda}     &  \text{Gaussian and Fock states,}\\ [10pt]
\frac{1}{\sqrt\lambda}   & \text{Coherent states.}
\end{array}
\right. 
\end{align*}

\subsection{Markovianity of $\Lambda_t$}

Notions of markovianity involve the concept of complete positivity. This was originally formulated by Stinespring in the context of $C^*$-algebras \cite{Stinespring} but it can be naturally extended to $*$-algebras of bounded operators on some Hilbert space that do not have the $C^*$ property, nor a unity\footnote{Note however that some results, like the Stinespring representation theorem, generally do not hold in these settings (see the counterexample in \cite{Attal}, p.~43).}, like the Hilbert-Schmidt operators $\T$. By definition, a linear map $L$ on $\T$ is called {\it completely positive} if $L_n \equiv L\otimes{\rm Id}_n$ is a positivity preserving map on $\T \otimes \mathcal B(\cx^n)=\mathcal T_2(L^2(\rx^{dN})\otimes \cx^n)$ for any $n\ge 1$, where $\mathcal B(\cx^n)$ are the (bounded) linear operators on $\cx^n$ and ${\rm Id}_n$ is the identity map on $\mathcal B(\cx^n)$. That is, 
\begin{align}
\text{$L$ completely positive on $\T$} \overset{{\rm def}}{ \Longleftrightarrow}  \big\{ X\in\T\otimes \mathcal B(\cx^n), X\ge 0\  \Rightarrow \ L_n(X)\ge 0\}. 
\end{align}
Here and throughout, an operator $X$ on a Hilbert space $\mathcal H$ is called positive, written $X\ge0$, if and only if $\langle f,Xf\rangle\ge 0$ for all $f\in\mathcal H$.

\begin{lem}
Fix $t\in\rx$. If $D_t(\vx,\vy)$ is continuous, then $\Lambda_t$ \eqref{defL} is completely positive on $\T$.  
\end{lem}

This result follows from Proposition \ref{prop_maps}, because by definition \eqref{decof} $D_t(\vx,\vy)$ is bounded and positive definite in the sense of \eqref{posdefkernel}.

According to \cite{divisibility}, different notions of (non-)Markovianity of a dynamical map $\Lambda_t$ are distinguished based on the properties of the map $V(t,s)$, defined by the relation
$$
\Lambda_t R = V(t,s)\Lambda_s R,\qquad \forall R\in\T, \, t\ge s\ge 0.
$$
The map $V(t,s)$ exists if the inverse $\Lambda_s^{-1}$ is well defined on $\T$ for all $t\ge s\ge 0$. Then
\begin{equation}
\label{m38}
V(t,s) = \Lambda_t\Lambda_s^{-1},\qquad t\ge s\ge 0.
\end{equation}
When $V(t,s)$ is well defined on $\T$ then $\Lambda_t$ is called:
\begin{itemize}
    \item {\it Markovian}, or {\it CP-divisible}, if $V(t,s)$ is completely positive
    \item {\it weakly non-Markovian}, if $\Lambda_t$ is {\it P-divisible} but not {\it CP-divisible}, i.e.~if $V(t,s)$ is positive (positivity preserving on $\T$) but not completely positive
    \item {\it essentially non-Markovian} if there exist $t,s$ such that $V(t,s)$ is not positive.
\end{itemize}

\noindent
Assume that $D_s(\vx,\vy)$ does not vanish for any $s,\vx,\vy$ and set
\begin{equation}
\label{Q}
Q_{t,s}(\vx,\vy):= \frac{D_t(\vx,\vy)}{D_s(\vx,\vy)},\quad \mbox{for all $t\ge s\ge 0$.}
\end{equation}
In accordance with \eqref{m38} we define the map $V(t,s)$ by the action (recall the notation $O_S$ given before \eqref{defL}),
\begin{equation}
\label{propagator}
   V(t,s) O_S = O_{Q_{t,s}S}.
\end{equation}
The map $V(t,s)$ leaves $\T$ invariant provided $Q_{t,s}\in L^\infty(\rx^{dN}\!\times\rx^{dN})$. 

\begin{thm}
\label{thm_markov}
Suppose that for all $t\ge s \geq 0$, the  function $Q_{t,s}$ is continuous and bounded in $(\vx,\vy)\in\rx^{dN}\times \rx^{dN}$.  The following statements are equivalent: 
\begin{itemize}
\item[(a)] $\Lambda_t$ is CP-divisible
\item[(b)] $\Lambda_t$ is P-divisible
\item[(c)] $Q_{t,s}(\vx,\vy)$ is a positive definite kernel (cf. \eqref{posdefkernel})
\end{itemize} 
\end{thm}
This theorem generalizes previous results obtained for pure dephasing models in finite-dimensional open quantum systems, notably Proposition 4.1 in \cite{Lonigro22} (see also similar results for a different model \cite{Chru23}). Theorem \ref{thm_markov} shows in particular that the notions of P-divisibility and CP-divisibility coincide in our model, so that the dynamics is either Markovian or essentially non-Markovian. Moreover, the distinction between the two cases is completely determined by the positivity of a kernel. 

\medskip

Markovianity properties are related to  monotonicity of the decoherence function by the following result.  
\begin{cor}
\label{cor2}
Suppose that $Q_{t,s}$ is bounded and continuous, as in Theorem \ref{thm_markov}. A necessary condition for $\Lambda_t$ to be CP-divisible is that $|D_t(\vx,\vy)| \leq |D_s(\vx,\vy)|$ for $t\geq s \geq 0$, that is, the decoherence function is monotonically decreasing as a function of time.
\end{cor}
We give a proof of Corollary \ref{cor2} in Section \ref{sec:proofs}. Checking that $Q_{t,s}(\vx,\vy)$ is a positive definite kernel may not be easy, as monotonicity of $t\mapsto | D_t(\vx,\vy) |$ is not sufficient (see for instance \cite{Marcantoni-Merkli-2}). On the other hand, when this monotonicity fails the dynamics is essentially non-Markovian.
\medskip

\begin{example}
\begin{itemize}
    \item[-]  Gaussian reservoir states $\omega_\r$ produce Markovian  $\Lambda_t$ because (see \eqref{gausdecfun})
    $$
    Q_{t,s}(\vx,\vy)= D_{\sqrt{t^2-s^2}}(\vx,\vy)
    $$ 
    is a positive definite kernel. Nevertheless, the dynamics $\Lambda_t$ is not a semigroup in $t$, because $D_{t+s}(\vx,\vy)\neq D_t(\vx,\vy) D_s(\vx,\vy)$ (see \eqref{defL} and \eqref{gausdecfun}).  This is a qualitative different feature from the weak coupling regime where the dynamics is given by a Markovian semigroup. 
\item[-] Coherent states $\omega_\r$ produce Markovian dynamics. Indeed ({\it cf.}~\eqref{gausdecfun}), 
    $$
    \sum_{k,l=1}^n \overline{\xi_k} \xi_l \, Q_{t,s}(\vx_k,\vx_l)= \sum_{k,l=1}^n \overline{\eta_k}\eta_l \, e^{-\frac14 (t^2 -s^2) (G(\vx_k)-G(\vx_l))^2 \, \|g\|^2_{L^2}} \geq 0,
    $$
    where we defined $\eta_k \equiv \xi_k e^{i (t -s) G(\vx_k) {\rm Im}\langle \alpha,g\rangle}$ and the positivity follows from the positivity of the Gaussian kernel. As for the Gaussian case, the semigroup property does not hold. 
    
\item[-] For Fock states $\omega_\r$ the decoherence function \eqref{gausdecfun} has zeroes because the Laguerre polynomials do. Thus $Q_{t,s}$ is not bounded and we cannot apply Theorem~\ref{thm_markov} for the study of markovianity, even though Theorem~\ref{thm:new2n} holds true and the effective dynamics $\Lambda_t$ is well-defined.
\end{itemize}
\end{example}

\subsubsection{Invariant domains of $\Lambda_t$}

Complex valued functions $S\in L^2(\rx^{dN}\!\times\rx^{dN})$ are in isometrically isomorphic correspondence with the Hilbert-Schmidt operators on $L^2(\rx^{dN})$, denoted $\mathcal T_2(\rx^{dN})$. Given $S$, the associated operator is the integral operator with kernel $S(\vx,\vy)$, which we denote by $O_S$. We have $\|S\|_{L^2}=\|O_S\|_{\rm HS} = \sqrt{{\rm tr}(O_S)^*O_S}$ (Hilbert-Schmidt norm). Let $S\in L^2(\rx^{dN}\!\times\rx^{dN})$ be such that for all integers $n\ge 1$, all $\xi_k\in\cx$, $\vx_k\in\rx^{dN}$, $k=1,\ldots,n$ we have 
\begin{equation}
\label{posdefkernel}
\sum_{k,l=1}^n \overline{\xi_k} \xi_l S(\vx_k,\vx_l) \ge 0.
\end{equation}
Then $S(\vx,\vy)$ is called a {\it positive definite kernel}. If $S(\vx,\vy)=\overline{S(\vy,\vx)}$ for all $\vx,\vy$, then $S(\vx,\vy)$ is called a {\it hermitian kernel}. If $S(\vx,\vy)$ is continuous in $\vx$ and $\vy$ then it is called a {\it continuous kernel}. We introduce the following operator spaces:
\begin{itemize}
\item[$\mathcal K$] $:=\{O_S : S(\vx,\vy)\in L^2(\rx^{dN}\!\times \rx^{dN}) \mbox{\ is a continuous kernel}\}$
\item[$\mathcal S$] $:=\{O_S :  S(\vx,\vy)\in L^2(\rx^{dN}\!\times \rx^{dN}) \mbox{\ is a continuous, hermitian and non-negative kernel}\}$
\item[$\mathcal S_0$] $:=\{O_S :  O_S\in\mathcal S \mbox{\ is trace-class}\}$
\end{itemize}
We have $\mathcal S_0\subset\mathcal S\subset\mathcal K\subset\mathcal T_2(\rx^{dN})$ and operators in $\mathcal S_0$ with unit trace are density matrices on $L^2(\rx^{dN})$. 

\begin{prop}
\label{prop:two}
Fix $t\ge 0$ and suppose that  $D_t(\vx,\vy)$, \eqref{decof} is continuous in $\vx,\vy$. Then $\Lambda_t$, \eqref{defL} leaves each of $\mathcal K$, $\mathcal S$ and $\mathcal S_0$ invariant. Moreover, $\Lambda_t$ is trace-preserving on $\mathcal S_0$. 
\end{prop}

{\bf Proof of Proposition \ref{prop:two}.} 
The invariance of $\mathcal K$ is immediate from the continuity and boundedness of $D_t(\vx,\vy)$. Next, $D_t(\vx,\vy)=\overline{D_t(\vy,\vx)}$ and so $\mathcal D_t S$ is a continuous, hermitian kernel if $S$ is. Let $S$ be a non-negative kernel and let $\xi_k\in\cx$, $\vx_k\in \rx^{dN}$, $k=1,\ldots,n$. Then
\begin{equation}
\sum_{k,l=1}^n\overline{\xi_k}\xi_l [\mathcal D_t S](\vx_k,\vx_l)  = \sum_{k,l=1}^n\overline{\xi_k}\xi_l D_t(\vx_k,\vx_l) S(\vx_k,\vx_l) = \langle \vec{\xi}, (A*B) \vec\xi\rangle,
\label{m31}
\end{equation}
where $\vec\xi\in\cx^n$ has components $\xi_k$, $A$ and $B$ are the $n\times n$ matrices with entries $A_{kl}=D_t(\vx_k,\vx_l)$ and $B_{kl}=S(\vx_k,\vx_l)$ and $\langle \cdot,\cdot\rangle$ is the Euclidean inner product of $\cx^n$. Here, $A*B$ is the Hadamard (entry-wise) product of $A$ and $B$. Both $A$ and $B$ are  non-negative definite matrices | which follows from the non-negativity of the kernels $D_t$ and $S$, by considering \eqref{m31} for one of the kernels alone. The Hadamard product of two non-negative definite matrices is non-negative definite (see {\it e.g.} \cite{Petz} Lemma 3.3 or \cite{Serre}). It follows that the right side of \eqref{m31} is $\ge 0$. Hence $\mathcal D_t S$ is a non-negative kernel. This shows that $\mathcal S$ is invariant under $\Lambda_t$. 

To show the invariance of $\mathcal S_0$ we use the following result given in \cite{Gohberg-Krein}, page 114 | see also \cite{Castro+}. If  $T(\vx,\vy): \rx^{dN}\times\rx^{dN}\rightarrow\cx$ is a continuous, hermitian and non-negative kernel, then the non-negative integral operator $O_T$ on $L^2(\rx^{dN})$ given by the kernel $T(\vx,\vy)$ is trace-class if and only if $\int T(\vx,\vx)d\vx<\infty$, and in this case the latter integral equals ${\rm tr}(T)$. We apply this to our setting. Let $S$ be a kernel such that $O_S\in\mathcal S_0$. Then 
$$
 {\rm tr}(O_S) = \int S(\vx,\vx)d\vx = \int S(\vx,\vx)D_t(\vx,\vx)d\vx,
$$
where the last equality is due to $D_t(\vx,\vx)=1$ for all $\vx$. By the same result from \cite{Gohberg-Krein} mentioned above, the right side equals ${\rm tr}\big(O_{\mathcal D_t S}\big)$. This concludes the proof of Proposition \ref{prop:two}. \hfill $\qed$
\bigskip

\subsection{Spatial localization of macroscopic quantum objects}
\label{sec:spatloc}

The spatial decoherence described by our results explains how localization of quantum objects in position space happens, and the fact that macroscopic objects localize faster than few-particle systems.

Consider a quantum object described by a system of $N$ quantum particles in $d$ spatial dimensions. If $N$ is very large then we shall call the object macroscopic. 
Let $\psi_1(\vx),\psi_2(\vx)\in L^2(\rx^{dN})$ be two wave functions of the object, where $\vx=(x_1,\ldots,x_N)\in \rx^{dN}$, $x_j\in \rx^d$, localized in compact disjoint subsets of the physical (position) space $I_1^N, I_2^N \subset \rx^{dN}$.  The wave function $\psi_1$ describes a state of the object where all its particles (and hence the object itself) is localized in $I_1\subset \rx^d$, and similarly for $\psi_2$ in $I_2$. The superposition 
\begin{align}
\label{39-2}
\psi(\vx) = \frac{1}{\sqrt 2}\big(\psi_1(\vx)+\psi_2(\vx)\big)
\end{align}
describes the prototypical {\it delocalized} state of the object, also called a Schr\"odinger cat state \cite{Joosetal, Haroche-Raimond, SchlossBook}.
The associated density matrix $\rho=|\psi\rangle\langle\psi|$  has the integral kernel
\begin{align}
\rho(\vx,\vy)=\psi(\vx)\overline{\psi(\vy)},
\end{align}
which is supported in $(\vx,\vy)\in (I_1^N\cup I_2^N)\times (I_1^N\cup I_2^N)$. The object is interacting with an environment  characterized by its decoherence function \eqref{decof}. We assume that the wave-functions $\psi_1(\vx)$ and $\psi_2(\vx)$ are continuous and that the decoherence function is continuous in $(\vx,\vy)$ for all $t$. According to Theorem \ref{thm:new2n}, \eqref{1m1} the object's density matrix at time $t$ | for the purpose of averages of integral (non-diagonal) operators $A$ | is approximated by $\Lambda_{\lambda t}(\rho)$ having the integral kernel ({\rm cf.}~\eqref{defL})
\begin{align}
\label{41}
\big[\Lambda_{\lambda t}&(\rho)\big](\vx,\vy) = D_{\lambda t}(\vx,\vy) \rho(\vx,\vy)=\frac12 D_{\lambda t}(\vx,\vy)\Big(\psi_1(\vx)+\psi_2(\vx)\Big)\Big(\overline{\psi_2(\vy)} +\overline{\psi_2(\vy)}\Big). 
\end{align}
Each particle of the object couples independently to the environment, 
\begin{align}
G(\vx)= \sum_{j=1}^N G_1(x_j),
\end{align}
for a continuous potential $G_1:\rx^d\rightarrow\cx$. We say that the potential $G_1$ resolves (or, distinguishes) the disjoint regions $I_1$ and $I_2$ if the values of $G_1$ on $I_1$ are all larger (or smaller) by some $\delta>0$, than all the values of $G_1$ on $I_2$ | in other words, if there is a $\delta>0$ such that
\begin{align}
\label{42}
\min_{x\in I_1,\, y\in I_2} \big(G_1(x)-G_1(y)\big)\ge  \delta\quad\text{or}\quad \max_{x\in I_1,\, y\in I_2} \big(G_1(x)-G_1(y)\big)\le  -\delta.
\end{align}
For such $G_1$ we have
\begin{align}
\label{Gbound}
\min_{\vx\in I_1^N,\, \vy\in I_2^N} \big| G(\vx)-G(\vy)|=\min_{\vx\in I_1^N,\, \vy\in I_2^N} \Big| \sum_{j=1}^N G_1(x_j)-G_1(y_j)\Big|\ge  N\delta .
\end{align}
Generically, the reservoir decoherence function  decays in time, and it does so uniformly in $\vx\in I_1^N$ and $\vy\in I_2^N$ provided \eqref{Gbound} holds (see for instance \eqref{gausdecfun}). Then we can define the reservoir decoherence time $\tau_\r$, determined by
\begin{align}
\label{43}
\max_{\vx\in I_1^N,\, \vy\in I_2^N} |D_t(\vx,\vy)|<\!\!<1,\qquad t\ge \tau_\r.    
\end{align}
The decoherence time depends on $N$. As the time parameter appears multiplied by $G(\vx)-G(\vy)$ in the decoherence function \eqref{decof}, we have the scaling
\begin{align}
\label{1/N}
\tau_\r\propto \frac1N.
\end{align}
It follows from \eqref{41} and \eqref{43} that 
\begin{align}
\label{44}
\sup_{\vx,\vy\in \rx^{dN}}\Big|[\Lambda_{\lambda t}(\rho)](\vx,\vy)-[\Lambda_{\lambda t}(\rho_{\rm MIX})](\vx,\vy)\Big|<\!\!<1,\qquad \lambda t\ge \tau_\r,
\end{align}
where $\rho_{\rm MIX}$ is the density matrix of the mixed state
\begin{align}
\rho_{\rm MIX} = \frac12 \big(|\psi_1\rangle\langle\psi_1| +|\psi_2\rangle\langle \psi_2|\big),
\end{align}
describing an ensemble of two {\it localized} states where one has equal probability of finding the object in $I_1$ or in $I_2$ (with no correlation between the two disjoint regions). Then
\begin{align}
\label{44.1}
\Lambda_{\lambda t}(\rho) \approx \frac12\Big(\Lambda_{\lambda t}(|\psi_1\rangle\langle\psi_1|) +\Lambda_{\lambda t}(|\psi_2\rangle\langle\psi_2|)\Big),\qquad \lambda t\ge \tau_\r.
\end{align}
Each of the components $|\psi_j\rangle\langle\psi_j|$ of $\rho_{\rm MIX}$, $j=1,2$, stays localized within their original support $I_j$, under the evolution $\Lambda_{\lambda t}$. Indeed, if the object is in the state $\Lambda_{\lambda t}(|\psi_j\rangle\langle\psi_j|)$, then the probability of finding it localized within a (Borel) set $L\subset\rx^{dN}$ is ($D_{\lambda t}(\vx,\vx)=1$), 
\begin{align} 
\label{46}
{\rm tr}\big(\{\Lambda_{\lambda t}(|\psi_j\rangle\langle\psi_j|)\}\chi_L\big) = \langle \psi_j,\chi_L\psi_j\rangle,
\end{align}
where $\chi_L$ is the characteristic function of $L$. Hence the support of the state $|\psi_j\rangle\langle\psi_j|$ does not change under the action of $\Lambda_{\lambda t}$. The relation  \eqref{44.1} tells us that $\Lambda_{\lambda t}(\rho)$ is close to a mixture of two {\it spatially localized} states. The mixture is fundamentally different from the delocalized state \eqref{39-2} before the interaction. It is an ensemble of localized object states, with a fifty percent chance each of finding the object localized in either $I_1$ or $I_2$. We conclude that the state of the object becomes spatially localized after interacting with the environment for a duration 
\begin{align} 
t_{\rm loc}=\tau_\r/\lambda.
\end{align}  
Due to the scaling \eqref{1/N}, macroscopic systems localize much quicker than those with a low number of particles. The spatial extension of the objects which become localized by the contact with the environment is determined by the resolution capacity of the interaction potential $G(\vx)$: If $G_1(x)$ varies considerably over a typical length $\ell_{\rm loc}$ then \eqref{42} holds for ${\rm dist}(I_1,I_2)\ge \ell_{\rm loc}$, so originally delocalized objects become localized on a spatial scale $\ell_{\rm loc}$.

We point out that the approximation \eqref{41} is valid for $\lambda t\le a$, see \eqref{1m1}. The localization happens within such a window of time, as indeed the condition $\lambda t\ge\tau_\r\propto 1/N$ in \eqref{44} is compatible with $\lambda t\le a$.
\medskip

\begin{example}
\begin{itemize}
\item[-] For small size regions $I_1$ and $I_2$ a sufficiently regular function $G_1$ is approximately constant on each of them. Suppose two different such constant values. Then $\Lambda_{\lambda t}(\rho_{\rm MIX})= \rho_{\rm MIX}$ (because in this case $D_t(\vx,\vy)=1$ for $\vx, \vy \in I_1^N$ or $\vx, \vy \in I_2^N$).

\item[-] Take $\psi$ of the form \eqref{39-2} and take $G_1$ continuous, constant equal to $g_j$ on $I_j$, $j=1,2$, with $\delta =|g_1-g_2|>0$. Let $\omega_\r$ be equilibrium state at temperature $T=1/\beta\ge0$. The decoherence function is,
\begin{align}
D_t(\vx,\vy) = e^{-\frac14t^2 (G(\vx)-G(\vy))^2 \langle g, \coth(\beta\omega/2)g\rangle}.
\end{align}
Then for $\vx \in I_1^N$ and $\vy \in I_2^N$ (or viceversa) we have $D_t(\vx,\vy)\le e^{-(t N\delta\alpha)^2}$, where $\alpha^2=  \frac14\langle g, \coth(\beta\omega/2)g\rangle$. The decoherence time is $\tau_\r=\frac{1}{N\delta\alpha}$.
\end{itemize}
\end{example}

\subsection{On the definition of the dynamics}
\label{sec:defdyn}

In the physics literature reservoirs are  described as a discrete (or finite) set of oscillators for which it is easier to perform calculations. A `continuous mode limit' is taken at later stages for quantities derived from the discrete model. In this approach it is difficult to control simultaneously the errors emerging from multiple approximations (markovian, long time, weak/strong coupling, discrete modes...). The algebraic description of the reservoir as a state over the Weyl algebra $\mathcal W_\r$, which we take here, dispenses with the need of a continuous mode approximation. However, a purely algebraic definition of the dynamics by an automorphism on the algebra of observables $\mathcal B(\h_\s)\otimes \mathcal W_\r$ for Bosonic systems is not in harmony with physically motivated interactions. Namely, on physical grounds, the Heisenberg dynamics of an observabel $O$ should be described by an interacting Hamiltonian $H$ via $e^{itH} O e^{-itH}$ | however physically motivated $H$ are such that  $e^{itH} O e^{-itH}$ does not belong to the algebra. 

Typically this problem is circumvented by fixing a reference state $\rho_\s\otimes\omega_{\rm ref}$, where $\omega_{\rm ref}$ is a convenient reservoir state, such as an equilibrium state. Then the interacting dynamics is {\it defined} by a unitary group, generated by a so-called Liouville operator $H$, in the Gelfand-Naimark-Segal (GNS) Hilbert space of $\rho_\s\otimes\omega_{\rm ref}$ \cite{Haag,BR,JaksicPillet, BFSRTE, MBSPRL,MMAOP,MMQI,MM22}. Assuming that $H_0$ is a self-adjoint generator of the unitarily implemented uncoupled dynamics, the interaction part is then accounted for by adding to $H_0$ the interaction operator containing the field operator $\varphi(g)$ (see \eqref{c1}) represented on the GNS Hilbert space, yielding the full Liouvillian $H$. The definition of the dynamics becomes then, on the mathematical level, dependent on the {\it folium} of the reference state (the collection of all states representable by density matrices on the GNS Hilbert space of $\rho_\s\otimes\omega_{\rm ref}$). Within the (huge!) class of regular states  | for which field operators exist by definition | and with an implementable uncoupled dynamics by $H_0$, the form of the Hamiltonian is {\it the same}, given by \eqref{c1}. This provides a certain universality of the approach. 

As described here, the setup hinges on the implementability of the free dynamics by a unitary group on the GNS Hilbert space. The implementability is guaranteed in particular for $\omega_{\r}$ which are stationary with respect to the free reservoir dynamics. This follows from the uniqueness of the GNS representation \cite{BR}. One could do without the condition of implementability of the reservoir dynamics. Namely, the full dynamics of an observable $A\equiv A\otimes\bbbone_\s$ can be defined by a Dyson series relative to the non-interacting dynamics, so the series only involves commutators with operators $e^{itH_\s}Ge^{-it H_\s}\otimes \varphi(e^{i\omega t}g)$ and the implementation $\varphi(e^{i\omega t}g)=e^{itH_\r}\varphi(g) e^{itH_\r}$ is not needed. Nevertheless, for the ease of the analysis, we assume implementability here. 

Finally we mention that instead of describing the Weyl algebra as an abstract $C^*$-algebra, for conceptual ease sometimes it is viewed as represented on Fock space (which corresponds to the GNS Hilbert space of the Gaussian states with covariance $\Omega=\bbbone$). Mathematically  speaking it does not make sense to define a KMS state on the algebra represented as operators on Fock space (KMS states are  defined on the $C^*$ Weyl algebra and they have Hilbert space representations which are not unitarily equivalent to the Fock representation  for nonzero temperature \cite{Takesaki}). Nevertheless, all algebraic properties solely based on the commutation relations, are independent of the Hilbert space representation | for example the commutator of field operators  $[\varphi(f),\varphi(g)]=i{\rm Im}\langle f,g\rangle$ is the same for all regular representations of the Weyl algebra (even though the operators are represented differently for different states). Wiewing the algebra as acting on Fock space to begin with does not lead to mistakes as long as only algebraic properties are involved  \cite{Marcantoni-Merkli, Marcantoni-Merkli-2}.

\medskip

Let $\omega_\r$ be a state on $\mathcal W_\r$ which is regular and for which the dynamics is implemented as in~\eqref{impl}. We define the operator
\begin{equation}
\label{Kop}
K\equiv K(\lambda) = H_\r +\lambda G\otimes \varphi(g),
\end{equation}
acting on $\h_{\s\r}=\h_\s \otimes \h_\r$, where $\h_\s=L^2(\rx^{dN})$ and $\h_\r$ is the GNS Hilbert space of $\omega_\r$. Let $H$ be given as in \eqref{c1}, acting on $\h_{\s\r}$. We obtain a Dyson series expansion for $A\in\mathcal B(\h_\s)$,
\begin{equation}
\label{r1}
e^{itH} (A\otimes\bbbone_\r) e^{-it H}
= A(t) + \sum_{n\ge 1}i^n\int_{0\le t_n\le\cdots\le t_1\le t} \big[ H_\s(t_n),\cdots \big[H_\s(t_2),[H_\s(t_1), A(t)] \big]  \cdots\big],
\end{equation}
where the integral is over $t_1,\ldots,t_n$ and
\begin{equation}
\label{r2}
A(t) = e^{itK} (A\otimes\bbbone_\r)e^{-itK}.
\end{equation}
The series converges in the operator norm of $\mathcal B(\h_{\s\r})$, uniformly in $\lambda$ for each fixed $t\ge 0$. The choice $K$, \eqref{r2} for the comparison dynamics in the Dyson series \eqref{r1} is suitable for the following reason. In the case of a finite-dimensional system, $\dim\h_\s<\infty$, the operator $A(t)$, \eqref{r2} 
is a finite linear combination of operators of the form $A_\s\otimes W(f)$ where $A_\s$ is an operator on $\h_\s$, see Lemma 1 of \cite{Marcantoni-Merkli}. Therefore the operator
\begin{equation}
\label{Bop}
B_{t,t_1,\ldots,t_n;A} := \big[ H_\s(t_n),\cdots \big[H_\s(t_2),[H_\s(t_1), A(t)] \big]  \cdots\big]
\end{equation}
is also such a finite linear combination and then $\rho\otimes\omega_\r$ (any state $\rho$ of $\s$) can be applied to \eqref{Bop} and one defines
\begin{equation}
\rho\otimes \omega_\r\big(e^{itH} (A\otimes\bbbone_\r) e^{-it H}\big)
:= \sum_{n\ge 0}i^n\int_{0\le t_n\le\cdots\le t_1\le t} \rho\otimes \omega_\r\big(B_{t,t_1,\ldots,t_n;A}\big).
\label{formaldyn2}
\end{equation}
The right hand side converges uniformly in $\lambda$ and so the limit $\lambda\rightarrow\infty$ can be taken term-wise in the series and integral. This is the approach of \cite{Marcantoni-Merkli}. In the present work we have $\dim\h_\s=\infty$, and the finite linear combinations expressing the multi-commutator \eqref{Bop} turn into integrals over operators of the form $A_\s\otimes W(f)\in {\mathcal B}(\h_\s)\otimes \mathcal W_\r$. More precisely, we show in Proposition \ref{propr1} that the operator $A(t)$, \eqref{r2} has a (reservoir operator valued) integral kernel $[A(t)](\vx,\vy)=e^{i\lambda^2 \Phi(t,\vx,\vy)} A(\vx,\vy) W(\lambda g_{t,\vx,\vy})$, where $\Phi$ is a real phase, $A(\vx,\vy)$ is the intergral kernel of the operator $A$ and $W$ is the Weyl operator smoothed out with a function $\lambda g_{t,\vx,\vy}$. The  operator $B_{t,t_1,\ldots,t_n;A}$, \eqref{Bop}, is a linear combination of products $X_1(s_1)\cdots X_n(s_n)$ (at different times $s_j$ and $X=H_\s$, $A$) having kernel
$$
[X_1(s_1)\cdots X_n(s_n)](\vx,\vy) = \int [X_1(s_1)](\vx,\vw_1)[X_2(s_1)](\vw_1,\vw_2) \cdots [X_n(s_n)](\vw_{n-1},\vy),
$$
where the integral is over the variables $\vw_1,\ldots,\vw_{n-1}\in\rx^{dN}$. Then we set
\begin{align}\label{defomega}
\omega_\r\big( & [X_1(s_1)\cdots X_n(s_n)](\vx,\vy)\big)\nonumber\\
&:=  \int \omega_\r\Big([X_1(s_1)](\vx,\vw_1)[X_2(s_1)](\vw_1,\vw_2) \cdots [X_n(s_n)](\vw_{n-1},\vy)\Big),
\end{align}
which is well defined as the argument of $\omega_\r$ in the right-hand side is a product of $n$ Weyl operators. By linearity, $\omega_\r([B_{t,t_1,\ldots,t_n;A}](\vx,\vy))$ is well defined for all $\vx,\vy$ and we set
\begin{equation}
\label{formaldyn3}
\rho\otimes \omega_\r\big(B_{t,t_1,\ldots,t_n;A}\big) := \int_{\rx^{dN}\!\times\rx^{dN}} \rho(\vx,\vy) \omega_\r\big([B_{t,t_1,\ldots,t_n;A}](\vy,\vx)\big) d\vx d\vy.
\end{equation}
Here, $\rho(\vx,\vy)$ is the integral kernel of $\rho$. Combining \eqref{formaldyn2} with \eqref{formaldyn3} leads us to the {\bf definition of the dynamics},
\begin{eqnarray}
\label{defdyn}
\langle A\rangle_t &:=& 
\rho\otimes\omega_\r\big(e^{itH} (A\otimes\bbbone_\r)e^{-itH}\big)\nonumber\\
&:=& 
\sum_{n\ge 0}i^n\int_{0\le t_n\le\cdots\le t_1\le t} \ \int_{\rx^{2dN}} \rho(\vx,\vy) \omega_\r\big([B_{t,t_1,\ldots,t_n;A}](\vy,\vx)\big),
\end{eqnarray}
where the integrals are over all the $t_j$ and over $\vx,\vy$.

\section{Proofs}
\label{sec:proofs}

{\bf Expressing the trace of an integral operator.} Recall that given an integral operator $A$ with integral kernel $A(\vx,\vy)$ we let $A^+$ be the integral operator with kernel $|A(\vx,\vy)|$. 

\begin{prop}
\label{prop:trace}
Let $\rho$ be a finite rank density matrix on $\h_\s=L^2(\mathbb R^{dN})$ with integral kernel $\rho(\vx,\vy)$ and let $A$ be a bounded operator on $\h_\s$ with integral kernel $A(\vx,\vy)$.
Suppose that 
\begin{equation}
\label{tr0}
\rho(\vx,\vy) A(\vy,\vx) \in L^1(\rx^{dN}\!\times\rx^{dN}).
\end{equation}
Then we have
\begin{equation}
\label{tr1}
{\rm tr}(\rho A) = \int \rho(\vx,\vy) A(\vy,\vx) d\vx d\vy.
\end{equation} 
Moreover, if $A^+$ is a bounded operator then \eqref{tr0} and \eqref{tr1} hold.
\end{prop}

{\bf Proof of Proposition \ref{prop:trace}.}
To show \eqref{tr1} we write $\rho =\sum_{j=1}^J p_j |\psi_j\rangle\langle\psi_j|$, 
where the probabilities $0\le p_j\le 1$ add up to $1$ and the $\psi_j$ are an orthonormal family in $\h_\s$. The integral kernel of $\rho$ is,
\begin{equation}
\label{tr2}
\rho(\vx,\vy) = \sum_{j=1}^J p_j\psi_j(\vx)\overline{\psi_j(\vy)}.
\end{equation}
We augment the family $\{\psi_j\}_{j=1}^J$ to an orthonormal basis $\{\psi_\ell\}_{\ell\in \mathbb N}$ of $\h_\s$. Then
\begin{align}
{\rm tr}(\rho A) &= \sum_{\ell\in\mathbb N} \langle \psi_\ell, \rho A\psi_\ell\rangle= \sum_{\ell\in\mathbb N} \sum_{j=1}^J p_j \langle \psi_\ell, (|\psi_j\rangle\langle\psi_j|) A\psi_\ell\rangle = \sum_{j=1}^J p_j \langle \psi_j,  A\psi_j\rangle \nonumber\\
& = \sum_{j=1}^J p_j \int \overline{\psi_j(\vy)} (A\psi_j)(\vy) d\vy= \int \Big[ \int A(\vy,\vx)\sum_{j=1}^J p_j \psi_j(\vx)\overline{\psi_j(\vy)}  d\vx\Big] d\vy\nonumber\\
&= \int \Big[ \int \rho(\vx,\vy) A(\vy,\vx)  d\vx\Big] d\vy.
\label{tr3}
\end{align}
Due to \eqref{tr0} the Fubini-Tonelli theorem implies that the iterated integral in \eqref{tr3} equals \eqref{tr1}. Moreover, by \eqref{tr4} and \eqref{tr2} we have that if $\|A^+\|<\infty$ then \eqref{tr0} is satisfied, and so \eqref{tr1} holds. \hfill $\qed$
 
\bigskip

\noindent
{\bf Operator valued integral kernels.} Let $B\in\mathcal B(\h_\s\otimes\h_\r)$ be such that there is a reservoir operator valued function $Z(\vx,\vy)$ so that for any $\psi_\s,\psi'_\s\in\h_\s$ and $\psi_\r,\psi'_\r\in\h_\r$ we have 
\begin{equation}
\label{defk}
\langle \psi_\s\otimes\psi_\r, B\psi_\s'\otimes\psi_\r'\rangle = \int \overline{\psi_\s(\vx)}  \psi_\s'(\vy) \langle\psi_\r, Z(\vx,\vy) \psi_\r'\rangle d\vx d\vy. 
\end{equation}
Then we call $Z(\vx,\vy)$ the (reservoir operator valued) kernel of $B$ and we write $B(\vx,\vy)=Z(\vx,\vy)$.

\begin{prop}
\label{propr1}
Let $O\in\mathcal B(\h_\s)$ have integral kernel $O(\vx,\vy)$. Then the kernel of $O(t)$, \eqref{r2}
is given by
\begin{equation}
\label{13.13}
\big[ e^{itK}(O\otimes\bbbone_\r)e^{-it K}](\vx,\vy) = O(\vx,\vy) Y(t,\vx,\vy),
\end{equation}
where 
\begin{equation}
Y(t,\vx,\vy) = e^{i\lambda^2\Phi(t,\vx,\vy)} W\big(\lambda g_{t,\vx,\vy}\big)
\label{c33}
\end{equation}
is a phase times a Weyl operator, with
\begin{eqnarray}
g_{t,\vx,\vy}(k) &=& \big(G(\vx)-G(\vy)\big)\frac{e^{i\omega(k) t}-1}{i\omega(k)}g(k),\label{cc21}\\
\Phi(t,\vx,\vy) &=&  -\tfrac 12 \big(G(\vx)-G(\vy)\big)^2\,  {\rm Im} \langle\frac{e^{i\omega t}-1-i\omega t}{\omega^2}g,g\rangle.
\end{eqnarray}
\end{prop}

{\bf Proof of Proposition \ref{propr1}.} The equality \eqref{13.13} follows from ({\it cf.}~\eqref{defk}) $e^{itK}\psi_\s(\vx)\otimes\psi_\r = \psi_\s(\vx)\otimes e^{it(H_\r+\lambda G(\vx)\varphi(g))}\psi_\r$ and from the following formula \eqref{c19}, which is obtained by using the polaron transformation (or a Trotter product argument), see for instance \cite{Marcantoni-Merkli} Lemma 1: For any $\gamma_\ell,\gamma_r\in\rx$, any $g\in L^2(\rx^3,d^3k)$ s.t.~$g/\omega\in L^2(\rx^3,d^3k)$, any $t\in\rx$,
\begin{equation}
\label{c19}
e^{it(H_\r+\gamma_\ell\varphi(g))}e^{-it(H_\r+\gamma_r\varphi(g))} = e^{-\frac i2 (\gamma_\ell-\gamma_r)^2 {\rm Im} \langle\frac{e^{i\omega t}-1-i\omega t}{\omega^2}g,g\rangle} W\Big(\frac{e^{i\omega t}-1}{i\omega}(\gamma_\ell-\gamma_r)g\Big).
\end{equation}
This completes the proof of Proposition \ref{propr1}.\hfill $\qed$

\subsection{Proof of Proposition \ref{prop:intop}} 

Both $\mathcal I_+$ and $\mathcal V$ are linear spaces. Let $S,T\in\mathcal I_+$. Then
$$
\big| [ST](\vx,\vy)\big| = \big|\int S(\vx,\vw) T(\vw,\vy) d\vw\big|\le \int [S^+](\vx,\vw) [T^+](\vw,\vy) d\vw = [ S^+T^+](\vx,\vy).
$$
Hence for any $\psi,\phi\in\h_\s$,
\begin{align}
\int \big|\psi(\vx) [ST](\vx,\vy)\phi(\vy)\big| d\vx d\vy &\le \int |\psi(\vx)|\,  [S^+T^+](\vx,\vy) \,|\phi(\vy)| d\vx d\vy\nonumber\\
& = \big\langle|\psi|, S^+T^+ |\phi|\big\rangle \le \|\psi\|\, \|S^+\|\, \|T^+\|\, \|\phi\|.
\end{align}
Therefore by \eqref{tr4} products of operators in $\mathcal I_+$ belong to $\mathcal I_+$. Clearly, the set $\mathcal V$ is also closed under multiplication. Now let $T\in\mathcal I_+$, $V\in\mathcal V$. Then $TV$ has integral kernel $[TV](\vx,\vy) = T(\vx,\vy) V(\vy)$,
so $|[TV](\vx,\vy)|\le \|V\|_\infty |T(\vx,\vy)|$ and therefore by \eqref{tr4}, $\|(TV)^+\|<\infty$. In the same way, $\|(VT)^+\|<\infty$.  It follows that $\mathcal I_+$ is invariant under the multiplication (from the left and from the right) by elements of $\mathcal V$. Therefore any product of elements of $\mathcal I_+$ and $\mathcal V$ is an element either of $\mathcal I_+$ or of $\mathcal V$. Thus $\mathcal A$ is an algebra. It is also clear that if $T\in\mathcal I_+$, then $T^*\in\mathcal I_+$, and $V\in\mathcal V$ implies $V^*\in\mathcal V$. So $\mathcal A$ is a $*$-algebra. \hfill $\qed$

\subsection{Proof of Theorem \ref{thm1.0}} 
Given the definition \eqref{defdyn} of the dynamics, we have to analyze the limit
\begin{eqnarray}
\label{dd1}
\lefteqn{\lim_{\lambda\rightarrow\infty} \rho\otimes\omega_\r\big(e^{itH} (A\otimes\bbbone_\r)e^{-itH}\big)}\nonumber\\
&=& 
\lim_{\lambda\rightarrow\infty}\sum_{n\ge 0}i^n\int_{0\le t_n\le\cdots\le t_1\le t} \int_{\rx^{dN}\!\times\rx^{dN}} \rho(\vx,\vy) \omega_\r\big([B_{t,t_1,\ldots,t_n;A}](\vy,\vx)\big).
\end{eqnarray}
The operators $B_{t,t_1,\ldots,t_n;A}$ are given in \eqref{Bop}. We first show that the series and the integrals converge uniformly in $\lambda\in\rx$, so that the limit can be taken termwise.
Consider one of the $2^n$ terms resulting from expanding out the multicommutator in \eqref{Bop}, 
\begin{equation}
\label{1.31}
T_{t,s_1,\ldots,s_n} = H_\s(s_1)\cdots H_\s(s_L) A(t) H_\s(s_{L+1})\cdots H_\s(s_n),
\end{equation}
where the $s_j$ are a permutation of the $t_j$ and $0\le L\le n$. For an operator $O\in\mathcal V$ we have $O(t)=O\otimes\bbbone_\r$ (see \eqref{r2}) because $O\otimes\bbbone_\r$ and $e^{itK}$ commute. For $O\in \mathcal I_+$, $O(t)$ has an (operator valued) integral kernel given in Proposition \ref{propr1}, \eqref{13.13}. We synthesize these two cases by adopting the notation
\begin{align}
[O(t)](\vx,\vy) = \left\{
\begin{array}{ll}
W(\vx)\delta(\vx-\vy) & \text{$O\in\mathcal V$ is the multiplication by the function $W(\vx)$}\\[2ex]
O(\vx,\vy)Y(t,\vx,\vy) & \text{$O\in \mathcal I_+$, see \eqref{13.13}.}
\end{array}
\right.
\label{58}
\end{align}
In \eqref{58}, $\delta(\vx-\vy)$ is the `delta function' which is here defined by the  rule $\int F(\vy)\delta(\vx-\vy)d\vy = F(\vx)$, regardless of the regularity of $F$. In other words, whenever $\delta(\vx-\vy)$ shows up in a (multiple) integral including the integration over $\vy$, we eliminate $\vy$ in the integrand by setting it equal to $\vx$  and we integrate over the remaining variables.

We have $H_\s=T+U$ with $T\in\mathcal I_+$ and $U\in\mathcal V$, so according to \eqref{58}, the integral kernel of $H_\s$ is 
\begin{align}
\label{01}
h(\vx,\vy) = T(\vx,\vy) + U(\vx)\delta(\vx-\vy),
\end{align}
where role of $\delta$ is interpreted as mentioned above. In the same vein, $A=a+b$ with $a\in\mathcal I_+$ and  $b\in\mathcal V$, so
\begin{align}
\label{02}
A(\vx,\vy) = a(\vx,\vy)+b(\vx)\delta(\vx-\vy).
\end{align}
We define $H_\s^+$ and $A^+$ as the integral operators having the kernels $|T(\vx,\vy)| + |U(\vx)|\delta(\vx-\vy)$ and  $|a(\vx,\vy)|+|b(\vx)|\delta(\vx-\vy)$, respectively. Therefore we have
\begin{equation}
\label{norm+}
\|H_\s^+\|\le \|T^+\| + \|U\|_\infty,\qquad \|A^+\|\le \|a^+\|+\|b\|_\infty.
\end{equation}
The integral kernel of $T_{t,s_1,\ldots,s_n}$ is,
\begin{eqnarray}
\lefteqn{T_{t,s_1,\ldots,s_n}(\vy,\vx)}\label{r3}\\
&=& \int h(\vy,\vw_1) h(\vw_1,\vw_2)\cdots h(\vw_{L-1},\vw_L) A(\vw_L,\vw_{L+1}) h(\vw_{L+1},\vw_{L+2})\cdots h(\vw_n,\vx)\nonumber\\
&&\times Y(s_1,\vy,\vw_1)\cdots Y(s_L,\vw_{L-1},\vw_L)
Y(t,\vw_L,\vw_{L+1}) Y(s_{L+1},\vw_{L+1},\vw_{L+2})\cdots Y(s_n,\vw_n,\vx), \nonumber
\end{eqnarray}
where the integral carries over all variables $\vw_j$, $j=1,\ldots,n$. Using the CCR \eqref{bog1} and the definition \eqref{c33}, we obtain for the product of the operators $Y$ in \eqref{r3},
\begin{eqnarray}
\lefteqn{Y(s_1,\vy,\vw_1)\cdots Y(s_n,\vw_n,\vx)} \nonumber \\
&=& e^{i \Phi'}W\Big(\lambda\Big[ g_{s_1,\vy,\vw_1}+ \sum_{j=2}^{L} g_{s_j,\vw_{j-1},\vw_j} + g_{t,\vw_L,\vw_{L+1}} +\sum_{j=L+1}^{n-1} g_{s_j,\vw_j,\vw_{j+1}} + g_{s_n,\vw_n,\vx}
\Big]\Big)
\label{c40.1}
\end{eqnarray}
where $\Phi'\in\rx$ is a phase depending on $\lambda$, all integration variables and all times. Since $|\omega_\r(W(\cdot))|\leq 1$ we have the estimate
\begin{eqnarray}
\lefteqn{\big|\omega_\r(T_{t,s_1,\ldots,s_n}(\vy,\vx))\big|}\nonumber\\
&\le& \int \big| h(\vy,\vw_1) h(\vw_1,\vw_2)\cdots h(\vw_{L-1},\vw_L) A(\vw_L,\vw_{L+1}) h(\vw_{L+1},\vw_{L+2})\cdots h(\vw_n,\vx)\big|\nonumber\\
&\le& \big[(H_\s^+)^L A^+ (H_\s^+)^{n-L}\big](\vy,\vx).
\label{39}
\end{eqnarray}

Therefore, as $\rho(\vx,\vy)=\sum_{j=1}^J p_j \psi_j(\vx)\overline{\psi_j(\vy)}$ (see \eqref{indmat}),
\begin{eqnarray}
\lefteqn{
\Big|\int_{\rx^{dN}\times\rx^{dN}} \rho(\vx,\vy) \omega_\r\big(T_{t,s_1,\ldots,s_n}(\vy,\vx)\big)\Big|}\nonumber\\
&\le&  \sum_{j=1}^J p_j \int |\psi_j(\vx)| \big[(H_\s^+)^L A^+ (H_\s^+)^{n-L}\big](\vy,\vx) |\psi_j(\vy)|\nonumber \\
&=&  \sum_{j=1}^J p_j \big\langle |\psi_j|, (H_\s^+)^L A^+ (H_\s^+)^{n-L} |\psi_j|\big\rangle\nonumber \\
&\le& \sum_{j=1}^J p_j \|\psi_j\|^2 \ \|H_\s^+\|^n \|A^+\| =    \|H_\s^+\|^n\, \|A^+\|.
\label{060}
\end{eqnarray}
As the upper bound \eqref{060} does not depend on the particular arrangement of the times $s_j$ and the value of $L$ in \eqref{1.31}, we obtain 
\begin{equation}
\label{061}
\Big| \int_{\rx^{dN}\times\rx^{dN}} \rho(\vx,\vy) \omega_\r\big([B_{t,t_1,\ldots,t_n;A}](\vy,\vx)\big)\Big| \le  2^n   \|H_\s^+\|^n\, \|A^+\|,
\end{equation}
where the norms on the right side satisfy \eqref{norm+}. It follows from the Weierstrass $M$-test and the Lebesgue Dominated Convergence Theorem that 
\begin{eqnarray}
\label{dd2-1}
\lefteqn{\lim_{\lambda\rightarrow\infty} \rho\otimes\omega_\r\big(e^{itH} (A\otimes\bbbone_\r)e^{-itH}\big)}\nonumber\\
&=& 
\sum_{n\ge 0}i^n\int_{0\le t_n\le\cdots\le t_1\le t} \lim_{\lambda\rightarrow\infty}\int_{\rx^{dN}\times\rx^{dN}} \rho(\vx,\vy) \omega_\r\big([B_{t,t_1,\ldots,t_n;A}](\vy,\vx)\big).
\end{eqnarray}
Now $\rho(\vx,\vy)\in L^2(\rx^{dN}\!\times \rx^{dN})$ due to condition (A1). The limit can be taken inside the last integral (Lebesgue Dominated Converge Theorem),
\begin{eqnarray}
\label{dd2}
\lefteqn{\lim_{\lambda\rightarrow\infty} \rho\otimes\omega_\r\big(e^{itH} (A\otimes\bbbone_\r)e^{-itH}\big)}\nonumber\\
&=& 
\sum_{n\ge 0}i^n\int_{0\le t_n\le\cdots\le t_1\le t} \int_{\rx^{dN}\times\rx^{dN}} \rho(\vx,\vy) \lim_{\lambda\rightarrow\infty}\omega_\r\big([B_{t,t_1,\ldots,t_n;A}](\vy,\vx)\big).
\end{eqnarray}
To evaluate the limit we look at $\lim_{\lambda\rightarrow\infty}\omega_\r\big(T_{t,s_1,\ldots,s_n}(\vy,\vx)\big)$, where the operator $T_{t,s_1,\ldots,s_n}(\vy,\vx)$ is given in \eqref{r3} and \eqref{c40.1}. The expectation of \eqref{c40.1} in the Gaussian state $\omega_\r$, \eqref{c17} is given by (see also \eqref{cc21})
\begin{eqnarray}
\lefteqn{
\omega_\r\Big( e^{i\Phi'} W\Big(\lambda\Big[ g_{s_1,\vy,\vw_1}+ \sum_{j=2}^{L} g_{s_j,\vw_{j-1},\vw_j} + g_{t,\vw_L,\vw_{L+1}} +\sum_{j=L+1}^{n-1} g_{s_j,\vw_j,\vw_{j+1}} + g_{s_n,\vw_n,\vx}
\Big]\Big)\Big)}\nonumber\\
&&=
e^{i\Phi'} \exp -\frac14 \lambda^2 \Big\| \sqrt{\mathcal C} \Big\{\big(G(\vy)-G(\vw_1)\big)\frac{e^{i\omega s_1}-1}{i\omega}+\sum_{j=2}^L \big(G(\vw_{j-1})-G(\vw_j)\big)\frac{e^{i \omega s_j}-1}{i\omega} \nonumber\\
&&\qquad +\big(G(\vw_L)-G(\vw_{L+1})\big)\frac{e^{i\omega t}-1}{i\omega} +\sum_{j=L+1}^{n-1} \big(G(\vw_j)-G(\vw_{j+1})\big)\frac{e^{i \omega s_j}-1}{i\omega} \nonumber\\
&&\qquad + \big(G(\vw_n)-G(\vx)\big)\frac{e^{i\omega s_n}-1}{i\omega}\Big\} g \Big\|_{L^2}^2.
\label{c41}
\end{eqnarray}
The following result is shown in Lemma 2 of \cite{Marcantoni-Merkli}. Suppose the condition (A4). Then given arbitrary values 
$$
G(\vy), G(\vw_1),\ldots,G(\vw_n),G(\vx)\in\rx,
$$
not all equal, and given an arbitrary $t>0$, the norm $\|\cdots\|^2_{L^2}$ in \eqref{c41} is strictly positive for almost every $(s_1,\ldots,s_n)\in\rx^n$. Therefore \eqref{c41} converges to zero as $\lambda\rightarrow\infty$, almost everywhere in $(s_1,\ldots,s_n)\in\rx^n$. On the other hand, if 
\begin{equation}
\label{c42}
G(\vy)=G(\vw_1)=G(\vw_2)=\cdots =G(\vw_n)=G(\vx),
\end{equation}
then that norm vanishes and so does $\Phi'$. It follows that for almost every $(s_1,\ldots,s_n)\in\rx^n$,
\begin{align}
\lim_{\lambda\rightarrow\infty}
\omega_\r\Big( e^{i\Phi'} & W\Big( \lambda\Big[g_{s_1,\vy,\vw_1}+ \sum_{j=2}^{L} g_{s_j,\vw_{j-1},\vw_j} + g_{t,\vw_L,\vw_{L+1}} +\sum_{j=L+1}^{n-1} g_{s_j,\vw_j,\vw_{j+1}} + g_{s_n,\vw_n,\vx}
\Big]\Big)\Big)\nonumber\\
&= \ind_\Gamma(\vy,\vw_1) \ind_\Gamma(\vw_1,\vw_2)\cdots \ind_\Gamma(\vw_{n-1},\vw_n)\ind_\Gamma(\vw_n,\vx),
\label{c43}
\end{align}
where $\ind_\Gamma(\vx,\vy)$ is the indicator function \eqref{i2}. Combining \eqref{r3}, \eqref{c40.1} and \eqref{c43}, and recalling the definition \eqref{defomega}, we arrive at the following result: Fix $t>0$, then for almost all $(s_1,\ldots,s_n)\in\rx^n$, we have 
\begin{eqnarray}
\lefteqn{\lim_{\lambda\rightarrow\infty}\omega_\r\big(T_{t,s_1,\ldots,s_n}(\vy,\vx)\big)}\nonumber\\
&=& \int  h(\vy,\vw_1) h(\vw_1,\vw_2)\cdots h(\vw_{L-1},\vw_L) A(\vw_L,\vw_{L+1}) h(\vw_{L+1},\vw_{L+2})\cdots h(\vw_n,\vx)\nonumber\\
&& \quad \times \ind_\Gamma(\vy,\vw_1) \ind_\Gamma(\vw_1,\vw_2)\cdots \ind_\Gamma(\vw_{n-1},\vw_n)\ind_\Gamma(\vw_n,\vx)\nonumber\\
&=&  [(\dg H_\s)^L (\dg A) (\dg H_\s)^{n-L}](\vy,\vx),
\label{83}
\end{eqnarray}
where the definition of $\dg$ is given in \eqref{i1}. As $\|(\dg H_\s)^+\|$, $\|(\dg A)^+\|<\infty$,   Proposition \ref{prop:trace} and \eqref{83} give
\begin{equation}
\int_{\rx^{dN}\times\rx^{dN}} \rho(\vx,\vy) \lim_{\lambda\rightarrow\infty}\omega_\r\big(T_{t,s_1,\ldots,s_n}(\vy,\vx)\big) = {\rm tr}\big(\rho (\dg H_\s)^L (\dg A) (\dg H_\s)^{n-L}\big).
\end{equation}
Remembering the $T_{t,s_1,\ldots,s_n}$, \eqref{1.31}, as being the terms obtained from expanding the multi-commutator $B_{t,t_1,\ldots,t_n;A}$, \eqref{Bop}, we obtain ($n$-fold commutator)
\begin{equation}
\int_{\rx^{dN}\times\rx^{dN}} \rho(\vx,\vy) \lim_{\lambda\rightarrow\infty}\omega_\r\big([B_{t,t_1,\ldots,t_n;A}](\vy,\vx)\big) = {\rm tr}\big(\rho  \big[ \dg H_\s,\cdots \big[\dg H_\s,[\dg H_\s, \dg A] \big]  \cdots\big]\big).
\label{1.43}
\end{equation}
The integral over the time simplex in \eqref{dd2} now simply gives $\frac{(it)^n}{n!}$ and using that 
$$
\sum_{n\ge 0} \frac{(it)^n}{n!} \big[ \dg H_\s,\cdots \big[\dg H_\s,[\dg H_\s, \dg A] \big] = e^{it \dg H_\s} (\dg A) e^{-it\dg H_\s}
$$
(convergence in operator norm) we combine \eqref{dd2} and \eqref{1.43} to arrive at
\begin{equation}
\label{dd3}
\lim_{\lambda\rightarrow\infty} \langle A\rangle_t  = 
{\rm tr}\big(\rho\, e^{it \dg H_\s} (\dg A) e^{-it\dg H_\s}\big) = {\rm tr}\big(e^{-it \dg H_\s} \rho\, e^{it\dg H_\s} (\dg A)\big).
\end{equation}
This concludes the proof of Theorem \ref{thm1.0}.\hfill $\qed$

\subsection{A refined result and proof of Theorem \ref{thm:new2n}} 
\label{sec:ref}

The proof of Theorem \ref{thm:new2n} follows from a finer result given in Theorem \ref{thm:new2} below. We state and prove the latter, and then we give a proof of Theorem \ref{thm:new2n}.

\subsubsection{A refined result}

We make the following assumption.
\begin{itemize}
\item[{\bf (A6)}] There is a function $\theta_\r: [0,\infty)\times [0,\infty)\rightarrow [0,\infty)$ such that 
\begin{align}
\label{deftheta}
\big| \omega_\r\big( W(\xi g)\big)\big|+\Big| \omega_\r\Big( W\Big(\xi \frac{e^{i\omega t}-1}{i\omega t} g\Big)\Big)\Big|\le \theta_\r(|\xi|,t),\qquad \xi\in\rx,\  t\ge 0.
\end{align}
 Moreover $\forall t\ge  0$, $x\mapsto \theta(x,t)$ is monotone decreasing in $x\ge 0$ and $\lim_{x\rightarrow\infty}\theta_\r(x,t)=0$.
\end{itemize}
Given the functions $\theta_\r$ and $\mu_\r$ ({\rm cf.}~\eqref{mudef}) we define the function of $\epsilon>0$, $t\ge 0$, $\lambda\in\rx$,
\begin{align}
C_\epsilon(\lambda,t) &= \theta_\r(t\lambda^\epsilon,t) + \mu_\r(t^2\lambda^{\epsilon},t) + \frac13 t^3\lambda^{2\epsilon} \|\sqrt\omega g\|_{L^2}^2+2t\|H^+_\s\|e^{2t\|H_\s^+\|}.
\label{CA}
\end{align}

\begin{thm}[Temporal resolution of the decoherence process]
\label{thm:new2}
Assume (A0)-(A3).
\begin{itemize}
\item[\rm 1.] 
\begin{itemize}
    \item[\rm (i)]
For any $A\in\mathcal V$ a multiplication operator by a function $V(\vx)$ we have
\begin{align}
\label{m29.1}
 \big| \langle A\rangle_t
- {\rm tr}\big(\rho A\big) \big|\le 2t  \,e^{2t\|H_\s^+\|} \|V\|_\infty.
\end{align}
\item[\rm (ii)]  Suppose that (A5) and (A6) hold and that $\rho(\vx,\vy), D_t(\vx,\vy)$ are continuous in $(\vx,\vy)$, for all $t\ge 0$. Let $A\in\mathcal I_+$ be such that $A(\vx,\vy)$ is bounded and continuous in $(\vx,\vy)$. Then we have for all $\epsilon>0$, $t\ge 0$ and $\lambda\in\rx$,
\begin{align}
\label{m29'}
\Big| \langle A\rangle_t
- {\rm tr}\big(\Lambda_{\lambda t}(\rho)\, A\big) \Big|
\le \|A^+\| \,C_\epsilon(\lambda,t).
\end{align}
\end{itemize}
\item[\rm 2.] For $A\in\mathcal I_+$, even without the continuity or boundedness conditions of {\rm 1(ii)}, we have for all $\epsilon>0$, $t\ge 0$ and $\lambda\in\rx$,
\begin{align}
\label{m29}
\Big| \langle A\rangle_t
- \int_{\rx^{dN}\times\rx^{dN}} \rho(\vx,\vy) D_{\lambda t}(\vx,\vy) A(\vy,\vx)d\vx d\vy\Big|\le \|A^+\|\, C_\epsilon(\lambda,t).
\end{align}
\end{itemize} 
\end{thm}

As for the function $\mu_\r(x,t)$, we have explicit functions for $\theta_\r(x,t)$ for Gaussian, coherent and Fock states as presented in Example \ref{ex2.1} | compare also with Propostion \ref{prop2n}.

\begin{prop}
\label{prop2}
We have the following explicit expressions.
\begin{itemize}
\item[-] For a {\bf Gaussian state} with covariance $\mathcal C$, \eqref{c17} we can take
\begin{align}
\label{i10}
\theta_\r(x,t) = 2e^{-\frac14 x^2  \|\frac{e^{i\omega t}-1}{i\omega t}g\|^2_{L^2}}.
\end{align}

\item[-] The same \eqref{i10} can be taken for a {\bf coherent state} $\omega_\r(\cdot) = \langle W(\alpha)\Omega, (\cdot) W(\alpha)\Omega\rangle$.

\item[-] For a {\bf Fock state} $\omega_\r= \langle \Psi_N, (\cdot)\Psi_N\rangle$ having $N$ particles in the state $h\in L^2(\rx^3,d^3k)$ we can take
\begin{align}
\label{ratefock}
\theta_\r(x,t)&= 2\Big(2^N + \| h \|^{2N}_{L^2} \|g\|_{L^2}^{2N}\big( 1+ (16 N)^{\!2N} \Big\|\frac{e^{i\omega t}-1}{i\omega t}g \Big\|^{-4N}_{L^2} \big)\Big) e^{-\frac18 x^2 \|\frac{e^{i\omega t}-1}{i\omega t}g\|^{2}_{L^2}}.
 \end{align}
\end{itemize}
\end{prop}

{\bf Proof of Propositions \ref{prop2n} and \ref{prop2}. } For the Gaussian state we have \begin{align*}
\Big|\omega_\r \Big(W \Big(\xi \frac{e^{i\omega t}-1}{i\omega t}g \Big) \Big)\Big|=\exp -\frac14|\xi|^2 \Big\|\mathcal C^{1/2}\frac{e^{i\omega t}-1}{i\omega t}g \Big\|^2_{L^2}\le \exp-\frac14 |\xi|^2 \Big\|\frac{e^{i\omega t}-1}{i\omega t}g \Big\|^2_{L^2}
\end{align*}
by assumption (A3) and because $\mathcal C\ge\bbbone$ ({\it cf.} before \eqref{c17}). Moreover, $\| g \|^2_{L^2}$ is lower bounded by the $L^2$ norm in the last inequality (because $|e^{i\omega t} -1|\leq \omega t$). Therefore we can choose $\theta_\r(x,t) = 2\exp-\frac14 x^2  \|\frac{e^{i\omega t}-1}{i\omega t}g\|^2_{L^2}$. Next we have  \begin{align}
\Big|1-\omega_\r\big(W(\xi   g_t)\big)\Big|=\Big|1-e^{-\frac14 |\xi|^2\frac{1}{t^2} \|\mathcal C^{1/2}\big(\frac{e^{i\omega t}-1}{i\omega t}-1\big) g \|^2_{L^2}}\Big|\le \frac14 |\xi|^2\ \Big\|\mathcal C^{1/2}\Big(\frac{e^{i\omega t}-1}{i\omega t}-1\Big)\frac1t g \Big\|^2_{L^2}
\label{i32}
\end{align}
and we may take the right hand side to define $\mu_\r(|\xi|,t)$. If $\mathcal C$ is diagonal (multiplication operator of a function of $k$), then we can further use the estimate $|\frac{e^{i\omega t}-1}{i\omega t} -1|\le \frac12 |\omega|t$ to bound \eqref{i32} from above by $\frac{1}{16} |\xi|^2\ \|\mathcal C^{1/2}\omega g\|^2_{L^2}$.

For the coherent state we have $\omega_\r(W(\xi f))= e^{-\frac14 \xi^2\|f\|^2_{L^2}} e^{i\xi{\rm Im}\langle \alpha,f\rangle}$. The phase term is irrelevant for the left side of \eqref{deftheta} and so we can take the same $\theta_\r$ as for the Gaussian state with covariance $\mathcal C=\bbbone$. The bound $|1-e^{-x}e^{i\varphi}| \leq x + |\varphi|$ for $x \geq 0, \varphi \in \mathbb{R}$ gives ({\it cf.}~\eqref{i32}), 
\begin{align}
\Big|1-\omega_\r\big(W(\xi  g_t)\big)\Big|&\le \frac{1}{16} \xi^2\ \|\omega g\|^2_{L^2} +\frac12 |\xi|\, \|\alpha\|_{L^2}\| \omega g\|_{L^2}.
\label{i33}
\end{align}
and \eqref{ratecohern} follows.

For the Fock state we have $\omega_\r(W(\xi f))=
\langle\psi_N,W(\xi f)\psi_N\rangle
=
e^{-\frac14|\xi|^2\|f\|_{L^2}^2}
L_N(|\xi|^2\frac{|\langle h,f\rangle|^2}{2})$, where 
\begin{align}
\label{laguerre}
L_N(x)=\sum_{k=0}^N{N\choose k}\frac{(-1)^k}{k!} x^k
\end{align}
is the $N$-th Laguerre polynomial. We have $|L_N(x)|\le 2^N(1+|x|^N)$, so
\begin{align}
\Big|\omega_\r\Big(W \Big(\xi \frac{e^{i\omega t}-1}{i\omega t}g\Big)\Big)\Big|&\le 2^N\Big(1+2^{-N}|\xi|^{2N} \| h \|^{2N}_{L^2} \Big\| \frac{e^{i\omega t}-1}{i\omega t}g \Big\|^{2N}_{L^2}\Big) e^{-\frac14 |\xi|^2\|\frac{e^{i\omega t}-1}{i\omega t}g\|^{2}_{L^2}}\nonumber\\
&\le  2^N\big(1+2^{-N}|\xi|^{2N} \| h \|^{2N}_{L^2} \|g\|^{2N}_{L^2}\big) e^{-\frac14 |\xi|^2\|\frac{e^{i\omega t}-1}{i\omega t}g\|^2_{L^2}}
\label{bnd}
\end{align}
where we used $|\frac{e^{i\omega t}-1}{i\omega t}|=\frac1t|\int_0^t e^{i\omega x}dx|\le1$. We construct an upper bound which is monotone decaying in $|\xi|$ (as per definition of $\theta_\r$). Set temporarily $\alpha'=\frac14\|\frac{e^{i\omega t}-1}{i\omega t}g\|^{2}_{L^2}>0$. We have $\xi^{2N} = e^{2N \log(|\xi|)} \leq e^{2N |\xi|}\le e^{\alpha'\xi^2/2}$ for $|\xi| > 4N/\alpha'$. On the other hand, for $|\xi| \leq 4N/\alpha'$ one has $\xi^{2N} \leq (4N/\alpha' )^{2N}$. Therefore, in any case
\begin{equation*}
    \xi^{2N} \leq e^{\alpha'\xi^2/2} + \Big(\frac{4N}{\alpha'}\Big)^{\!2N} \leq e^{\alpha'\xi^2/2} \Big( 1+ \Big(\frac{4N}{\alpha'}\Big)^{\!2N} \Big)
\end{equation*}
and \eqref{bnd} yields
\begin{align*}
 \Big| \omega_\r \Big(W \Big(\xi \frac{e^{i\omega t}-1}{i\omega t}g \Big) \Big) \Big| &\le 2^N \Big(1+2^{-N} 
\| h \|^{2N}_{L^2} \|g\|_{L^2}^{2N}\Big( 1+ \Big(\frac{4N}{\alpha'}\Big)^{\!2N} \Big) e^{\alpha'\xi^2/2}\Big) e^{-\alpha'\xi^2}\nonumber\\
&\le \Big(2^N + \| h \|^{2N}_{L^2} \|g\|_{L^2}^{2N}\Big( 1+ \Big(\frac{4N}{\alpha'}\Big)^{\!2N} \Big)\Big) e^{-\alpha'\xi^2/2}.
\end{align*}
The same estimate applies to $|\omega_\r (W \big(\xi g ) ) |$ because $\| g \|_{L^2} \geq \|\frac{e^{i\omega t}-1}{i\omega t}g\|_{L^2}$ for any $t \geq 0$. This gives $\theta_\r$ in \eqref{ratefock}. Next,
\begin{align*} 
\Big|1- e^{-\frac14|\xi|^2\|f\|^2}
L_N\Big(|\xi|^2\frac{|\langle h,f\rangle|^2}{2}\Big)\Big| \leq \Big|1- e^{-\frac14|\xi|^2\|f\|^2}\Big| + 
\Big|1- L_N\Big(|\xi|^2\frac{|\langle h,f\rangle|^2}{2}\Big)\Big|
\end{align*}
We use the bound \eqref{i32} for the first term on the right side, with $f=(\frac{e^{i\omega t}-1}{i\omega t}-1\big)\frac 1t g$ and $\mathcal{C}=\bbbone$, while for the second one we use the bound $|1- L_N(x)| \leq 2^N |x| (1+ |x|^{N-1})$. Then
\begin{align*}
\Big|1-\omega_\r\big(W(\xi  g_t)\big)\Big|&\le \frac{1}{16} \xi^2\ \|\omega g\|^2_{L^2} + 2^N \frac{\xi^2}{8} \|\omega g\|^2_{L^2} \|h \|^2_{L^2} \Big( 1 + \Big(\frac{\xi^2}{8} \|\omega g\|^2_{L^2} \|h \|^2_{L^2} \Big)^{\!N-1}\Big) \\
& = \frac{1}{16} \xi^2\ \|\omega g\|^2_{L^2} \Big( 1 + 2^{N+1} \|h \|^2_{L^2} \Big( 1 + \Big(\frac{\xi^2}{8} \|\omega g\|^2_{L^2} \|h \|^2_{L^2} \Big)^{\!N-1}\Big) \Big).
\end{align*}
This shows \eqref{ratefockn}. The proof of Propositions \ref{prop2n} and \ref{prop2} is complete.\hfill \qed

\color{black}

\medskip

{\bf Proof of Theorem \ref{thm:new2}.} Set $\tau=\lambda t\ge 0$. For $\lambda>0$, using the identity $e^{itH}=e^{i\tau(H_\s/\lambda + K/\lambda)}$, we rewrite the Dyson series expansion \eqref{r1} with the substitutions $t \to \tau, H_\s \to H_\s/\lambda$ and $K \to K/\lambda$.
Explicitly, we have
\begin{align}
\label{Dyson}
e^{itH} (A&\otimes\bbbone_\r) e^{-it H}\nonumber\\
&= A'(\tau) +\sum_{n\ge 1}\frac{i^n}{\lambda^n}\int_{0\le \tau_n\le\cdots\le \tau_1\le \tau}  \big[ H'_\s(\tau_n),\cdots \big[H'_\s(\tau_2),[H'_\s(\tau_1), A'(\tau)] \big]  \cdots\big],
\end{align}
where the integral is over $\tau_1,\ldots,\tau_n$ and
\begin{equation}
\label{r2'}
X'(\tau) = e^{i\tau K'} (X\otimes\bbbone_\r)e^{-i\tau K'},\quad \mbox{with}\quad  K'= H_\r/\lambda+ G\otimes\varphi(g).
\end{equation}
For multiplication operators $X\in\mathcal V$ we have
\begin{equation}
\label{x'=x}
X'(\tau)=X
\end{equation}
because $A\otimes \bbbone_\r$ and $e^{-i\tau K'}$ commute. 
Therefore 
\begin{align}
\Big| \langle A\rangle_t &-\rho\otimes\omega_\r\big(A'(\tau)\big)\Big|\nonumber\\
&\le 
\sum_{n\ge 1}\frac{1}{\lambda^n}\int_{0\le \tau_n\le\cdots\le \tau_1\le \tau} \int_{\rx^{dN}\times\rx^{dN}} |\rho(\vx,\vy)|\big| \omega_\r\big([B'_{\tau,\tau_1,\ldots,\tau_n;A}](\vy,\vx)\big)\big|,
\label{dd4}
\end{align}
where $B'_{\tau,\tau_1,\ldots,\tau_n;A} = \big[ H'_\s(\tau_n),\cdots \big[H'_\s
(\tau_2),[H'_\s(\tau_1), A'(\tau)] \big]$.  
We show that the right side of \eqref{dd4} vanishes in the limit  $\lambda\rightarrow\infty$. The operator $B'$ can be expanded into a sum of $2^n$ terms by undoing the commutators. Each term is of the form 
\begin{equation}
\label{1.31.1}
T'_{\sigma_1,\ldots,\sigma_n,\tau} = H'_\s(\sigma_1)\cdots H'_\s(\sigma_L) A'(\tau) H'_\s(\sigma_{L+1})\cdots H'_\s(\sigma_n),
\end{equation}
where the $\sigma_j$ are a permutation of the $\tau_j$ and $0\le L\le n$. For $A\in \mathcal I_+$ we have 
\begin{eqnarray}
\lefteqn{\omega_\r\big([T'_{\sigma_1,\ldots,\sigma_n,\tau}](\vy,\vx)\big)}\label{exp'}\\
&=& \int h(\vy,\vw_1) h(\vw_1,\vw_2)\cdots h(\vw_{L-1},\vw_L) A(\vw_L,\vw_{L+1}) h(\vw_{L+1},\vw_{L+2})\cdots h(\vw_n,\vx)\nonumber\\
&&\!\!\!\!\!\!\!\!\!\!\!\!\!\!\!\times \omega_\r\Big(Y'(\sigma_1,\vy,\vw_1)\cdots Y'(\sigma_L,\vw_{L-1},\vw_L)
Y'(\tau,\vw_L,\vw_{L+1}) Y'(\sigma_{L+1},\vw_{L+1},\vw_{L+2})\cdots Y'(\sigma_n,\vw_n,\vx)\Big), \nonumber
\end{eqnarray}
where $Y'$ is obtained from $Y$, \eqref{c33}, by the replacement $t\rightarrow\tau$, $\omega(k)\rightarrow \omega(k)/\lambda$ and $g\rightarrow g/\lambda$, namely,
\begin{equation}
Y'(\tau,\vx,\vy) = e^{i\lambda^2\Phi'(\tau,\vx,\vy)} W\big(\lambda g'_{\tau,\vx,\vy}\big),
\label{c33-1}
\end{equation}
with
\begin{eqnarray}
g'_{\tau,\vx,\vy}(k) &=& [G(\vx)-G(\vy)]\frac{e^{i\omega(k) \tau/\lambda}-1}{i\omega(k)}g(k),\label{55}\\
\Phi'(\tau,\vx,\vy) &=&  -\tfrac 12 [G(\vx)-G(\vy)]^2\,  {\rm Im} \langle\frac{e^{i\omega \tau/\lambda}-1-i\omega \tau/\lambda}{\omega^2}g,g\rangle.
\label{56}
\end{eqnarray}
In \eqref{exp'} we have $|\omega_\r(\cdots)|\le 1$ and we obtain as in \eqref{061},
$$
\int_{\rx^{dN}\times\rx^{dN}} |\rho(\vx,\vy)| \big| \omega_\r\big([B'_{\tau,\tau_1,\ldots,\tau_n;A}](\vy,\vx)\big)\big| \le (2\|H_\s^+\|)^n \|A^+\|,
$$
where the operator norms satisfy \eqref{norm+}. The right side of \eqref{dd4} is thus bounded above by $\|A^+\|\sum_{n\ge 1} \frac{1}{\lambda^n}\frac{(2\|H^+_\s\|\tau)^n}{n!} =\|A^+\|\big(e^{2\|H_\s^+\|\tau/\lambda}-1\big)\le \|A^+\| \frac{2\|H^+_\s\|\tau}{\lambda} e^{2\|H_\s^+\|\tau/\lambda}$ so we obtain from \eqref{dd4},
\begin{align}
\Big| \langle A\rangle_t -\rho\otimes\omega_\r\big(e^{i\tau(H_\r/\lambda+G\otimes\varphi(g))} (A\otimes\bbbone_\r)e^{-i\tau (H_\r/\lambda +G\otimes\varphi(g))}\big)\Big| \le \|A^+\| \frac{2\|H^+_\s\|\tau}{\lambda} e^{2\|H_\s^+\|\tau/\lambda}.
\label{dd4.1}
\end{align}
For $A\in\mathcal V$ we have $e^{i\tau(H_\r/\lambda+G\otimes\varphi(g))} (A\otimes\bbbone_\r)e^{-i\tau (H_\r/\lambda +G\otimes\varphi(g))}=A\otimes\bbbone_\r$ and the proof of \eqref{m29.1} is complete at this stage. For $A\in\mathcal I_+$ we make further estimates. By Proposition \ref{propr1} we have for $A\in\mathcal I^+$,
\begin{align}
\label{--}
\big[ e^{i\tau(H_\r/\lambda +G\otimes\varphi(g))}(A\otimes\bbbone_\r)e^{-i\tau (H_\r/\lambda+G\otimes\varphi(g))}](\vx,\vy) &= A(\vx,\vy) Y'(\tau,\vx,\vy),\nonumber\\
\big[ e^{i\tau G\otimes\varphi(g)}(A\otimes\bbbone_\r)e^{-i\tau G\otimes\varphi(g)}](\vx,\vy) &= A(\vx,\vy) W\big(\tau[G(\vx)-G(\vy)]g\big),
\end{align}
with $Y'(\tau,\vx,\vy)$ given in \eqref{c33-1}. Then for $A\in\mathcal I^+$,
\begin{align}
\Big|\rho\otimes\omega_\r&\big(e^{i\tau(H_\r/\lambda+G\otimes\varphi(g))} (A\otimes\bbbone_\r)e^{-i\tau (H_\r/\lambda +G\otimes\varphi(g))}\big)-\rho\otimes\omega_\r\big(e^{i\tau G\otimes\varphi(g)} (A\otimes\bbbone_\r)e^{-i\tau  G\otimes\varphi(g)}\big)\Big|\nonumber\\
&\le \int_{\rx^{dN}\times\rx^{dN}} |\rho(\vx,\vy)| |A(\vx,\vy)| \Big| \omega_\r\Big(Y'(\tau,\vx,\vy)-W\big(\tau[G(\vx)-G(\vy)]g\big) \Big)\Big|\nonumber\\
&\le \|A^+\|\, \esssup_{\vx,\vy}\Big| \omega_\r\Big(Y'(\tau,\vx,\vy)-W\big(\tau[G(\vx)-G(\vy)]g\big) \Big)\Big|
\label{81}
\end{align}
(see also \eqref{060} for the last inequality). We now bound the supremum.

Fix $\epsilon>0$ and divide $\rx^{dN}\times\rx^{dN}$ into two disjoint sets
\begin{equation}
\label{i3}
\mathcal R_1(\lambda) = \big\{(\vx,\vy) : |G(\vx)-G(\vy)|\ge \lambda^{-1+\epsilon} \big\},\quad\mathcal R_2(\lambda) = \big\{(\vx,\vy) : |G(\vx)-G(\vy)|< \lambda^{-1+\epsilon} \big\}.
\end{equation}
We have for $(\vx,\vy)\in\mathcal R_1(\lambda)$,
\begin{align}
\Big| \omega_\r\Big(&Y'(\tau,\vx,\vy)-W\big(\tau[G(\vx) -G(\vy)]g\big) \Big)\Big|\nonumber\\
& \le \Big|\omega_\r\Big(W \Big(\lambda t [G(\vx)-G(\vy)] \frac{e^{i\omega t}-1}{i\omega t}g \Big)\Big)\Big| + \Big|\omega_\r\Big( W\Big(\lambda t [G(\vx)-G(\vy)] g\Big)\Big)\Big|\nonumber\\
&\le \theta_\r\big(\lambda t \lambda^{-1+\epsilon},t\big) = \theta_\r\big(t\lambda^\epsilon,t\big),
\label{110}
\end{align}
where we used \eqref{deftheta}. (If the set $\mathcal R_1$ is empty then we cosider $\theta_\r=0$.) Next we estimate the supremum in \eqref{81} for $(\vx,\vy)\in \mathcal R_2(\lambda)$.  Set temporarily \begin{align}
\label{112}
f'(k)=\lambda g'_{\tau,\vx,\vy}(k)\quad \text{and}\quad  f(k)=\tau[G(\vx)-G(\vy)]g(k),
\end{align}
both $f,f'\in L^2(\rx^3,d^3k)$. Recalling the definition of $Y'(\tau,\vx,\vy)$ \eqref{c33-1}, one has
\begin{align}\label{bound1}
\Big| \omega_\r&\Big(Y'(\tau,\vx,\vy)-W\big(\tau[G(\vx)-G(\vy)]g\big) \Big)\Big| \nonumber \\
&= \Big| \omega_\r\Big( e^{i\lambda^2\Phi'} W\big(f'\big) - W\big(f\big) \Big)\Big| 
= \Big| \omega_\r\Big(  W\big(f'\big) - e^{-i\lambda^2\Phi'}W\big(f\big) \Big)\Big| \nonumber \\
&= \Big| \omega_\r\Big(  W\big(f')\big[\bbbone - e^{i\Phi''}W(f-f')\big] \Big)\Big| \nonumber \\
&\leq \Big| \omega_\r\Big( \big[e^{i \Phi''} -1\big] W(f')W(f-f')\Big) \Big| + \Big| \omega_\r\Big( W(f')\big[\bbbone- W\big(f -f'\big) \big] \Big)\Big| \nonumber \\
&\le \big| e^{i \Phi''} -1 \big| \Big|\omega_\r\big(W(f)\big)\Big|+ \sqrt{\omega_\r\Big( \big[ W\big(f' -f\big) -\bbbone\big]^* \big[ W\big(f' -f\big) -\bbbone\big] \Big)},
\end{align}
where we used the Cauchy-Schwarz inequality to get the square root term. The phase is 
\begin{align*}
\Phi''\equiv -\lambda^2 \Phi' + \frac{1}{2}{\rm Im}\langle f', f \rangle = \frac12 \lambda^2t^2 [G(\vx) - G(\vy)]^2 \,{\rm Im} \Big\langle \left(\frac{e^{i\omega t}-1}{i\omega t} +\frac{e^{i\omega t}-1-i\omega t}{\omega^2 t^2}\right) g,g \Big\rangle.
\end{align*}
We estimate the first term on the right side of \eqref{bound1} from above by $|e^{i\Phi''}-1|$. Using that
\begin{equation}\label{expest}
   \frac{1-e^{-i\omega t}}{i\omega t} = 1+\frac{1}{i\omega t}\int_0^{\omega t}dx \int_0^xds \, e^{-is},\qquad \Big|{\rm Im} \Big\langle \frac{e^{i\omega t}-1}{i\omega t} g,g \Big\rangle \Big|\le \frac{t}{2}\|\sqrt\omega g\|_{L^2}^2,
\end{equation}
and also
\begin{align*}
\frac{e^{-i\omega t}-1+i\omega t}{\omega^2 t^2}  &= -\frac{1}{\omega^2 t^2}\int_0^{\omega t} dx\int_0^x dy \big(1-(1-e^{-i y})\big) \\
&= -\frac{1}{2} +\frac{i}{\omega^2t^2}\int_0^{\omega t} dx\int_0^x dy\int_0^y dz e^{-i z} ,
\end{align*}
we obtain 
\begin{align*}
\Big| {\rm Im} \Big\langle\frac{e^{i\omega t}-1-i\omega t}{\omega^2 t^2}g,g \Big\rangle \Big| \leq \frac{t}{3!}\|\sqrt\omega g\|_{L^2}^2.
\end{align*}
Therefore, for $(\vx,\vy)\in\mathcal R_2(\lambda)$,
\begin{align}
\big| e^{i \Phi''} -1 \big| &=  \Big| \int_0^{\Phi''} e^{ix}dx\Big|\le |\Phi''| \leq \frac12 \lambda^2t^3  \, |G(\vx)-G(\vy)|^2  \|\sqrt\omega g\|_{L^2}^2 \left( \frac{1}{2} + \frac{1}{3!} \right)\nonumber \\
&\le  \frac13 \lambda^2 t^3\lambda^{-2+2\epsilon}  \|\sqrt\omega g\|_{L^2}^2 = \frac13 t^3\lambda^{2\epsilon} \|\sqrt\omega g\|_{L^2}^2.
\label{i4}
\end{align}
This is an upper bound for the first term on the right side of \eqref{bound1} on the set $\mathcal R_2(\lambda)$. We bound the second term in \eqref{bound1} on this set as follows. We have 
\begin{align}
  &\sqrt{\omega_\r\Big( \big[ W\big(f' -f\big) -\bbbone\big]^* \big[ W\big(f' -f\big) -\bbbone\big] \Big)}\label{bound3}\\
  &= \sqrt{2-2{\rm Re}\, \omega_\r\big(W(f'-f)\big)  } = \sqrt{2{\rm Re}\,\Big(1- \omega_\r\big(W(f'-f)\big)\Big) } \leq \sqrt{2} \, \Big| 1- \omega_\r\big(W(f'-f)\big) \Big|^{1/2},
  \nonumber
\end{align}
where ({\it cf.} \eqref{112}) 
\begin{align}
\label{ff'}
f'(k) -f(k) = \lambda t^2[G(\vx)-G(\vy)]  \Big(\frac{e^{i\omega(k) t}-1}{i\omega(k)t}-1\Big)\frac1t g(k).
\end{align}
On $\mathcal R_2(\lambda)$ we have $\lambda t^2|G(\vx)-G(\vy)|\le t^2 \lambda^{\epsilon}$ and 
\begin{align}
\label{bound4}
\sqrt{\omega_\r\Big( \big[ W\big(f' -f\big) -\bbbone\big]^* \big[ W\big(f' -f\big) -\bbbone\big]}\Big)\le \mu_\r\big(t^2\lambda^{\epsilon},t\big),
\end{align}
where we used \eqref{mudef}. Combining \eqref{dd4.1}, \eqref{81}, \eqref{110} and \eqref{bound4} we find that for any $\epsilon>0$, any $A\in\mathcal I_+$,
\begin{align}
\Big| \langle A\rangle_t &-\rho\otimes\omega_\r\big(e^{i\tau G\otimes\varphi(g)} (A\otimes\bbbone_\r)e^{-i\tau G\otimes\varphi(g)}\big)\Big|\label{90}\\
& \le \|A^+\| \Big[ \theta_\r(t\lambda^\epsilon,t) + \mu_\r(t^2\lambda^{\epsilon},t) + \frac13 t^3\lambda^{2\epsilon} \|\sqrt\omega g\|_{L^2}^2+2t\|H^+_\s\|e^{2t\|H_\s^+\|}\Big].
\nonumber
\end{align}
This proves the estimate \eqref{m29'}.

Now we prove \eqref{m29}. It suffices to show that
\begin{align}
\label{93}
\int \rho(\vx,\vy) D_{\lambda t}(\vx,\vy) A(\vy,\vx) d\vx d\vy
 ={\rm tr}\big(\Lambda_{\lambda t}(\rho)\, A\big).
\end{align}
We use the following fact, shown in Theorem 3.1 of Brislawn, \cite{Brislawn}. If a trace class operator $T$ on $L^2(\rx^n)$ has a continuous kernel $T(\vx,\vy)$, then $\tr(T) = \int T(\vx,\vx)d\vx$. By Proposition \ref{prop:two}, $\Lambda_{\lambda t}(\rho)$ is trace class. 
As $A$ is a bounded operator, $\Lambda_{\lambda t}(\rho) A$ is trace class. Its integral kernel is
\begin{align}
[\Lambda_{\lambda t}(\rho) A](\vx,\vy) = \int \rho(\vx,\vw)D_{\lambda t}(\vx,\vw)A(\vw,\vy)d\vw.
\label{94}
\end{align}
We show that \eqref{94} is continuous in $\vx\in\rx^{dN}$. Let $\vx_n\rightarrow\vx$ in $\rx^{dN}$. By the continuity assumption, the integrand $\rho(\vx_n,\vw)D_{\lambda t}(\vx_n,\vw)A(\vw,\vy)$ converges to $\rho(\vx,\vw)D_{\lambda t}(\vx,\vw)A(\vw,\vy)$ pointwise, for all $\vw\in\rx^{dN}$. As $\rho$ is of the form \eqref{indmat} it is enough to consider $\rho(\vx,\vw)=\overline{\psi(\vx)}\psi(\vw)$ for $\psi\in L^1(\rx^{dN})\cap L^2(\rx^{dN})$. By the boundedness  $\sup_{\vw,\vy}|A(\vw,\vy)|=C<\infty$ and $|D_{\lambda t}(\vx_n,\vw)|\le 1$ we have $|\psi(\vw)D_{\lambda t}(\vx_n,\vw)A(\vw,\vy)|\le C|\psi(\vw)|$. Since $\psi\in L^1(\rx^{dN})$ we can apply the Lebesgue Dominated Convergence Theorem, 
$$
\lim_n\int\overline{\psi(\vx_n)}\psi(\vw)D_{\lambda t}(\vx_n,\vw)A(\vw,\vy)d\vw=\int\overline{\psi(\vx)}\psi(\vw)D_{\lambda t}(\vx,\vw)A(\vw,\vy)d\vw.
$$
This shows the continuity $\vx\mapsto [\Lambda_{\lambda t}(\rho)A](\vx,\vy)$. The continuity in $\vy$ is shown along the same lines. We can now apply Brislawn's result and the Fubini-Tonelli Theorem to conclude that \eqref{93} holds. This concludes the proof of Theorem \ref{thm:new2}. \hfill $\qed$

\subsubsection{Proof of Theorem \ref{thm:new2n}}

The upper bounds of \eqref{m29.1}-\eqref{m29} converge to zero provided (see \eqref{CA})
\begin{align}
\label{i6}
t\lambda^\epsilon\rightarrow \infty,\quad t^2\lambda^{\epsilon}\rightarrow 0, \quad t^3\lambda^{2\epsilon}\rightarrow 0,\quad t\rightarrow 0.
\end{align}
Pick any $\alpha>0$ and look at the time scale $t\propto \lambda^{-\alpha}$, $\lambda\rightarrow\infty$. Then \eqref{i6} holds for $\alpha<\epsilon<3\alpha/2$. Hence regardless of the scaling ($\alpha>0$) the upper bounds  in \eqref{m29.1}-\eqref{m29} will vanish in the limit  $\lambda\rightarrow\infty$. Next, with $\lambda t=\tau\lambda^{1-\alpha}$,
\begin{align}
{\rm tr}\big(\Lambda_{\lambda t}(\rho)A\big) = \int_{\rx^{dN}\!\times \rx^{dN}} \rho(\vx,\vy)D_{\tau\lambda^{1-\alpha}}(\vx,\vy) A(\vx,\vy)d\vx d\vy,
\end{align}
where $D_t(\vx,\vy)$ is the decoherence function \eqref{decof}. Due to (A5), \eqref{deftheta} we have $\lim_{t\rightarrow \infty} D_t(\vx,\vy)=0$ for all $(\vx,\vy)\not\in\Gamma$. Thus as $\lambda\rightarrow\infty$ the function $D_{\tau\lambda^{1-\alpha}}(\vx,\vy)$ converges to $\mathbf 1_\Gamma(\vx,\vy)$ almost everywhere for $\alpha<1$ while it converges to the constant function $1$ for $\alpha>1$. The Lebesgue Dominated Convergence Theorem then gives the statements 1.~and 2.

To prove 3.~we use the bound \eqref{81} and \eqref{bound1} (valid for all $\vx,\vy\in\rx^{dN}$, without using the decomposition into regions $\mathcal R_{1,2}$). The first term on the right side of \eqref{bound1} has the estimate ({\it cf.}~\eqref{i4} and use $|\omega_\r(W(f))|\le 1$)
\begin{align}
\label{i15}
\big|e^{i\Phi''}-1\big|\Big|\omega_\r\big(W(f\big)\Big|\le \frac13\lambda^2t^3 |G(\vx)-G(\vy)|^2 \|\sqrt\omega g\|^2_{L^2} \le \frac43 \tau^2 t\|G\|_\infty^2 \|\sqrt \omega g\|^2_{L^2}.
\end{align}
The square root term in \eqref{bound1} has the upper bound (see see \eqref{bound3}, \eqref{ff'}),
\begin{align} 
\label{i16}
\sqrt2 \Big|1-\omega_\r\big(W(f'-f)\big)\Big|^{1/2}\le \mu_\r\big(2\|G\|_\infty \tau t,t\big).
\end{align}  
 Combining \eqref{dd4.1}, \eqref{81}, \eqref{bound1} \eqref{i15} and \eqref{i16} we find that 
\begin{align*}
\Big| \langle A\rangle_t &-\rho\otimes\omega_\r\big(e^{i\tau G\otimes\varphi(g)} (A\otimes\bbbone_\r)e^{-i\tau G\otimes\varphi(g)}\big)\Big|\\
& \le \|A^+\| \Big[ \mu_\r\Big(2\|G\|_\infty \frac{\tau^2}{\lambda} ,\frac{\tau}{\lambda}\Big)+ \frac43 \frac{\tau^3}{\lambda} \|G\|_\infty^2 \|\sqrt \omega g\|^2_{L^2}+2\frac{\tau}{\lambda}\|H^+_\s\|e^{2\tau\|H_\s^+\|/\lambda}\Big],
\nonumber
\end{align*}
which gives \eqref{1m1}. This completes the proof of Theorem \ref{thm:new2n}.\hfill \qed
\medskip

\subsection{Proof of Theorem \ref{thm_markov}}

Recall that given a function $S(\vx,\vy)\in L^2(\rx^{dN}\!\times\rx^{dN})$ we denote by $O_S\in\mathcal T_2(L^2(\rx^{dN}))$ the associated Hilbert-Schmidt integral operator. For  $D\in L^\infty(\rx^{dN}\!\times\rx^{dN})$ let  $\mathcal D$ be the operator of multiplication by $D$. $\mathcal D$ is a bounded linear operator on $L^2(\rx^{2dN})$ and it induces a bounded linear map $L$ on $\mathcal T_2(L^2(\rx^{dN}))$,
\begin{equation}
\label{L}
L O_S = O_{\mathcal D S}. 
\end{equation}
We have $\|\mathcal DS\|_2\le \|D\|_\infty\|S\|_2$ ($L^p$ norms) and so the operator norm of $L$ satisfies $\|L\|\le \|D\|_\infty$. The proof of Theorem \ref{thm_markov} is based on the following result.

\begin{prop}
\label{prop_maps}
If the function $D(\vx,\vy)$ is bounded, continuous and hermitian, $D(\vx,\vy) = \overline{D(\vy,\vx)}$, then the following statements are equivalent
\begin{itemize}
\item[(a)] $L$ is completely positive 
\item[(b)] $L$ is positivity-preserving
\item[(c)] $D(\vx,\vy)$ is a positive definite kernel (in the sense of \eqref{posdefkernel})     
\end{itemize}  
\end{prop}

We give a proof of Proposition \ref{prop_maps} below. For now we use the proposition to prove Theorem \ref{thm_markov}. We first verify that $D_t(\vx,\vy)$ is hermitian. This follows directly from the definition \eqref{decof} together with $\overline{\omega_\r(X)}=\omega_\r(X^*)$ and $W(h)^*=W(-h)$. Next we show that $D_t(\vx,\vy)$ is a positive definite kernel. Using again \eqref{decof} gives  for any integer $n\ge 1$, any $\xi_i\in\cx$ and $\vx_i\in\rx^{dN}$,
\begin{equation*}
   \sum_{i,j=1}^{n} \overline{\xi_i} \xi_j \omega_\r( e^{i t G(\vx_j)\varphi(g)}e^{-i t G(\vx_i)\varphi(g)}) = \omega_\r (Y Y^*) \geq 0 ,
\end{equation*}
where $Y= \sum_j \xi_j e^{i t G(\vx_j)\varphi(g)} $. Applying Proposition \ref{prop_maps} with $L=\Lambda_t$ shows that $\Lambda_t$ is completely positive, which is the first assertion of Theorem \ref{thm_markov}. Next, recall that the map $V(t,s)$ is induced by the multiplication operator $Q_{t,s}(\vx,\vy)= D_t(\vx,\vy) / D_s(\vx,\vy)$ acting on integral kernels, see \eqref{propagator}. The equivalence of the statement (a)--(c) in Theorem \ref{thm_markov} is a direct consequence of Proposition~\ref{prop_maps}. (Take $L=V(s,t)$ and $D(\vx,\vy)=Q_{t,s}(\vx,\vy)$.)

The proof of Theorem \ref{thm_markov} is now complete, modulo the
\medskip

\noindent
{\bf Proof of Proposition \ref{prop_maps}.} The implication $(a) \Rightarrow (b)$ is immediate. We prove $(b) \Rightarrow (c)$ by showing the contrapositive. Assume that $D(\vx,\vy)$ is not a positive definite kernel | we show that then there exists an $O\in \mathcal T_2(\rx^{dN})$ with $O\ge 0$ (positive operator) and there exists an $f\in L^2(\rx^{dN})$, such that $\langle f, (LO)f\rangle<0$.

By assumption, there is an $n\ge1$ and there are $\xi_1,\ldots,\xi_n\in\cx$ and $\vx_1,\ldots,\vx_n\in\rx^{dN}$ such that (see \eqref{posdefkernel})
\begin{equation}
\label{x0}
\sum_{i,j=1}^n \overline{\xi_i}\xi_j D(\vx_i,\vx_j)<0.
\end{equation}
Without loss of generality, we consider that all $\vx_i$ are distinct\footnote{Indeed, if two coordinates coincide, say for instance $\vx_{n}= \vx_{n-1}$, one can rewrite \eqref{x0} as $\sum_{i,j=1}^{n-1} \overline{\eta_i}\eta_j D(\vx_i,\vx_j)$ where $\eta_i=\xi_i$ for $i=1,\ldots ,n-2$ and $\eta_{n-1}=\xi_{n-1}+\xi_{n}$.}. Let $\delta=\min_{i\neq j}\{\|\vx_i-\vx_j\|\}>0$ and for $0<\varepsilon<\delta/2$ choose  $f_\varepsilon(\vx) = \beta^{-1}\varepsilon^{-dN}\sum_{i=1}^n \xi_i \textbf{1}_{B(\vx_i,\varepsilon)}(\vx)$.  
Here, $B(\vx_i,\varepsilon)$ is the open ball in $\rx^{dN}$ centered at $\vx_i$ with radius $\varepsilon$ and $\beta=|B(0,1)|$ is the volume of the unit ball in $\rx^{dN}$.  The $\beta^{-1}\varepsilon^{-dN}\textbf{1}_{B(\vx_i,\varepsilon)}(\vx)$ approximates the delta function $\delta(\vx-\vx_i)$ as $\varepsilon\rightarrow 0$.  Let  $S(\vx,\vy)=(\sum_{i=1}^n\mathbf 1_{B(x_i,\delta)}(\vx))(\sum_{i=1}^n\mathbf 1_{B(x_i,\delta)}(\vy))\in L^2(\rx^{dN}\!\times \rx^{dN})$ be the characteristic function of the set $\cup_{i,j=1}^n B(\vx_i,\delta)\times B(\vx_j,\delta)\subset\rx^{dN}\!\times \rx^{dN}$. $O_\s$ is a positive (rank one) projection operator. Since $f_\varepsilon(\vx) \sum_{i=1}^n\mathbf 1_{B(x_i,\delta)}(\vx)=f_\varepsilon(\vx)$ we obtain,
\begin{equation}
\label{x1}
\langle f_\varepsilon, (L O_S)f_\varepsilon\rangle  = \beta^{-2}\varepsilon^{-2dN}\sum_{i,j=1}^n \overline{\xi_i}\xi_j \int \textbf{1}_{ B(\vx_i,\varepsilon)}(\vx)\textbf{1}_{ B(\vx_j,\varepsilon)}(\vy) D(\vx,\vy)d\vx d\vy.
\end{equation}
By the continuity of $D(\vx,\vy)$ and \eqref{x0} we have $\lim_{\varepsilon\rightarrow 0} \langle f_\varepsilon, (LO_S)f_\varepsilon\rangle=\sum_{i,j=1}^n\overline{\xi_i}\xi_j D(\vx_i,\vx_j)<0$. Therefore $\langle f_\varepsilon, (LO_S)f_\varepsilon\rangle<0$ for small enough $\varepsilon$. This completes the proof of (b) $\Rightarrow$ (c).

We now prove $(c) \Rightarrow (a)$. Let $0\le O\in\mathcal T_2\otimes \mathcal B(\cx^n)$ and define $L_n \equiv L \otimes \rm Id_n$ acting on such operators $O$. We want to prove that
\begin{align}
\label{want}
\langle f , L_n(O) f\rangle \geq 0 \quad \text{for any $f \in \mathcal{H}_n\equiv L^2(\rx^{dN})\otimes \cx^n$.}
\end{align}
Being a positive Hilbert-Schmidt operator, $O$ admits a spectral representation in terms of positive eigenvalues $\lambda_i$ and rank-one projections $P_i \equiv |\psi_i\rangle\langle\psi_i|$ onto orthonormal eigenfunctions $\psi_i \in \mathcal{H}_n$
\begin{align}
\label{OO}
O = \sum_{i=1}^\infty \lambda_i P_i, \quad\quad \sum_{i=1}^\infty \lambda_i^2 < \infty,
\end{align}
where the series expressing $O$ converges in Hilbert-Schmidt norm of operators on $\h_n$. The norm of $L_n$ as an operator on $\mathcal T_2(\h_n)$ satisfies $\| L_n \| = \| L \otimes {\rm Id}_n \| =\| L \| \leq \|D\|_\infty$, see the bound after \eqref{L} for the last inequality. Therefore $L_n$ is continuous and we obtain from \eqref{OO},  
\begin{align}
\label{OO1}
L_n(O)= \sum_{i=1}^\infty \lambda_i L_n(P_i).
\end{align}
The series \eqref{OO1} converges in the Hilbert-Schmidt norm. By the continuity of the inner product, 
\begin{equation}\label{scalar1}
    \big\langle f, L_n(O) f \big\rangle_{\mathcal{H}_n} = \sum_{i=1}^\infty \lambda_i \, \big\langle f, L_n (P_i) f\big\rangle_{\! \mathcal{H}_n}.
\end{equation}
Now, $P_i$ is a rank-one projection onto the function $\psi_i \in \mathcal{H}_n =L^2(\mathbb{R}^{dN}) \otimes \mathbb{C}^n$ that can be written\footnote{It may seem restrictive to consider only finite linear combinations of products but in fact this decomposition is exhaustive due to the finite dimension of $\mathbb{C}^n$. Indeed, given a orthonormal basis $\{f_k\}_{k=1}^\infty$ in $L^2$ and a orthonormal basis $\{ e_j \}_{j=1}^n$ in $\mathbb{C}^n$, the set of products $f_k \otimes e_j$ defines a orthonormal basis in $\mathcal{H}_n$. Therefore, a generic function $f \in \mathcal{H}_n$ can be written as $f= \sum_{j=1}^n\sum_{k=1}^\infty c_{jk} f_k \otimes e_j$ with complex coefficients such that $\sum_{j=1}^n\sum_{k=1}^\infty |c_{jk}|^2 <\infty$. This in particular implies that $\sum_{k=1}^\infty |c_{jk}|^2 <\infty$ for any fixed $j$ so that one can define the function $g_j \equiv \sum_{k=1}^\infty c_{jk}f_k$ belonging to $L^2$ and rewrite $f= \sum_{j=1}^n g_j \otimes e_j$. More generically, when considering two Hilbert spaces, one of which is finite dimensional, the algebraic tensor product is already complete and it coincides with the spatial tensor product.} as $\psi_i = \sum_{j=1}^n g^i_j \otimes e_j$ for some orthonormal basis $\{ e_j \}_{j=1}^n$ of $\mathbb{C}^n$ and $L^2$ functions $g^i_j$. Then $P_i = \sum_{j,k=1}^n |g^i_j\rangle\langle g^i_k| \otimes |e_j\rangle\langle e_k|$ and
\begin{equation*}
   L_n (P_i) = \sum_{j,k=1}^n L \big(|g^i_j\rangle\langle g^i_k| \big) \otimes |e_j\rangle\langle e_k|.
\end{equation*}
The operator $|g^i_j\rangle\langle g^i_k|$ is Hilbert-Schmidt with integral kernel $g^i_j(\vx)\overline{g^i_k(\vy)}$, so $L \big(|g^i_j\rangle\langle g^i_k| \big)(\vx,\vy)= g^i_j(\vx) D(\vx,\vy) \overline{g^i_k(\vy)}$. As explained after \eqref{scalar1} we write the function $f$ in \eqref{want} as $f= \sum_{k=1}^n f_k \otimes e_k$  and we obtain,
\begin{equation*}
    \big\langle f, L_n(P_i) f \big\rangle_{\! \mathcal{H}_n} = \sum_{j,k=1}^n \big\langle f_j , L \big(|g^i_j\rangle\langle g^i_k| \big) f_k \big\rangle_{L^2} = \sum_{j,k=1}^n \int \overline{f_j(\vx)} \Big[ \int g^i_j(\vx) D(\vx,\vy) \overline{g^i_k(\vy)} f_k(\vy) d\vy \Big] d\vx
\end{equation*}
The order of integration is immaterial by the Fubini-Tonelli theorem. Therefore we have 
\begin{equation}
\label{scalar2}
    \big\langle f, L_n(P_i) f \big\rangle_{\! \mathcal{H}_n} =  \int \overline{h^i(\vx)}  D(\vx,\vy) h^i(\vy) d\vx d\vy,
\end{equation}
where we have set $h^i= \sum_{k=1}^n \overline{g^i_k}f_k\in L^1(\rx^{dN})$ (integrability follows from $g^i_k,f_k\in L^2(\rx^{dN})$ and the Cauchy-Schwarz inequality). Finally we show that \eqref{scalar2} is positive. This follows from the positivity of the kernel $D(\vx,\vy)$. 
\begin{lem}
\label{lemma_posint'}
Let $D: \mathbb{R}^{dN} \times \mathbb{R}^{dN} \to \mathbb{C}$ be a bounded, continuous hermitian and positive definite function (c.f. \eqref{posdefkernel}). Then
\begin{equation}
\label{posint'}
\int \overline{f(\vx)} D(\vx,\vy) f(\vy) d\vx d\vy \geq 0\qquad \forall f \in L^1(\mathbb{R}^{dN},d\vx).
\end{equation}
\end{lem}
\noindent
This concludes the proof of Proposition \ref{prop_maps} because \eqref{want} follows from \eqref{scalar1} and \eqref{scalar2}. The complete positivity of $L$ is therefore demonstrated, modulo the proof of Lemma \ref{lemma_posint'}, which we present now.
\smallskip

{\bf Proof of Lemma \ref{lemma_posint'}.} Given $f\in L^1$ and $\epsilon>0$ there is a $g\in C_c(\rx^{dN})$, a continuous function with compact support, such that $\|f-g\|_{L^1}<\epsilon$. Writing $\langle f_1,Df_2\rangle = \int \overline{f_1(\vx)}D(\vx,\vy)f_2(\vy)d\vx d\vy$ we have,
\begin{equation}
\label{n1}
\big|\langle f,Df\rangle -\langle g, Dg\rangle\big| \le \|f-g\|_{L^1} \|D\|_\infty (\|f\|_{L^1}+\|g\|_{L^1}) \le C(f,D)\epsilon
\end{equation}
for a constant $C(f,D)$ not depending on $\epsilon$ (we have $\|g\|_{L^1}\le \|f\|_{L^1}+\epsilon\le 2\|f\|_{L^1}$ for $\epsilon$ small enough). Suppose we know that $\langle g,Dg\rangle\ge 0$ for all continuous, compactly supported $g$. Then \eqref{n1} implies $\langle f,Df\rangle \ge -C(f,D)\epsilon+\langle g,Dg\rangle\ge  -C(f,D)\epsilon$. As $\epsilon>0$ is arbitrary we get $\langle f,Df\rangle\ge 0$. So it is enough to show \eqref{posint'} for $f\in C_c(\rx^{dN})$.

Fix an $f\in C_c(\rx^{dN})$.  We are going to show that for arbitrary $\epsilon>0$,
\begin{equation}
\label{n0}
\langle f,Df\rangle\ge -C\epsilon
\end{equation}
for some constant $C$ (depending on $f,D$ but not on $\epsilon$). Then $\langle f,Df\rangle\ge0$ follows. Let $\epsilon>0$ be arbitrary, fixed. Denote by $\Omega\in\rx^{dN}$ the compact support of $f$. As $f$ is uniformly continuous on $\Omega$ and $D$ is uniformly continuous on $\Omega\times\Omega$, there is a $\delta>0$ such that 
\begin{eqnarray}
|f(\vx)-f(\vx')| &<& \epsilon/|\Omega|\label{com1}\\
|D(\vx,\vy)-D(\vx',\vy')|&<& \epsilon \label{com2}
\end{eqnarray}
for all $\vx,\vy,\vx',\vy'\in\Omega$ such that $\|\vx-\vx'\|<\delta$ and $\|(\vx,\vy)-(\vx',\vy')\|< \sqrt{2}\delta$. 
Take a finite cover of $\Omega$ by disjoint, measurable sets $\Omega_j$, $j=1,\ldots,M$, such that for any pair of points $\vx,\vx' \in \Omega_j$ one has $\|\vx-\vx'\|<\delta$. This implies that   $\|(\vx,\vy)-(\vx',\vy')\|< \sqrt{2}\delta$ if $(\vx,\vy), (\vx',\vy')\in \Omega_i\times \Omega_j$.
For each $j$ pick an  $\vx_j\in\Omega_j$ and define the simple function \begin{equation}
\label{h}
h(\vx) = \sum_{j=1}^M f(\vx_j) \mathbf 1_{\Omega_j}(\vx).
\end{equation}
Let $\vx\in\Omega$. Then $\vx\in\Omega_j$ for exactly one $j$ and $|f(\vx)-h(\vx)|=|f(\vx)-f(\vx_j)|<\epsilon/|\Omega|$ by \eqref{com1}. It follows that 
\begin{equation}
\label{n10}
\|f-h\|_{L^1} < \epsilon.
\end{equation}
As in \eqref{n1} this means that $\big|\langle f,Df\rangle -\langle h, Dh\rangle\big| < C_1(f,D)\epsilon$ for some constant $C_1$ not depending on $\epsilon$, and so,
\begin{equation}
\label{n2}
\langle f,Df\rangle \ge -C_1(f,D) \epsilon + \langle h, Dh\rangle.
\end{equation}
Next we show that $\langle h,Dh\rangle > -\epsilon(\epsilon+ \|f\|_{L^1})^2$. Then $\langle f,Df\rangle\ge 0$ follows from \eqref{n2} since $\epsilon>0$ is arbitrary. To estimate $\langle h,Dh\rangle$ we use the following expression
\begin{equation}
\label{n3}
\langle h,Dh\rangle = \sum_{i,j=1}^M \overline{f(\vx_i)} f(\vx_j) \int_{\Omega_i\times \Omega_j} D(\vx,\vy)d\vx d\vy.
\end{equation}
By \eqref{com2} we have $|\int_{\Omega_i\times \Omega_j} \big(D(\vx,\vy) -  D(\vx_i,\vx_j)\big) d\vx d\vy| < \epsilon |\Omega_i|\, |\Omega_j|$ and so we obtain from \eqref{n3} that 
\begin{equation}
\label{n5}
\Big| \langle h,Dh\rangle - \sum_{i,j=1}^M \overline{f(\vx_i)} f(\vx_j) |\Omega_i| |\Omega_j| D(\vx_i,\vx_j)\Big| < \epsilon \Big(\sum_{i=1}^M |\Omega_i|\, |f(\vx_i)|\Big)^2 = \epsilon \|h\|^2_{L^1} < \epsilon (\epsilon+\|f\|_{L^1})^2.
\end{equation}
We have used \eqref{n10} in the last inequality. Since $D(\vx,\vy)$ is a non-negative definite function (kernel), we have
\begin{equation}
\label{n6}
\sum_{i,j=1}^M \overline{f(\vx_i)} f(\vx_j) |\Omega_i| |\Omega_j| D(\vx_i,\vx_j)\ge 0.
\end{equation}
Using \eqref{n6} in \eqref{n5} yields
$\langle h,Dh\rangle > -\epsilon(\epsilon+ \|f\|_{L^1})^2$, which is the bound we were trying to get. This concludes the proof of Lemma \ref{lemma_posint'} and with that the proof of Proposition \ref{prop_maps}.\hfill $\qed$

\subsubsection{Proof of Corollary \ref{cor2}} 

Suppose $\Lambda_t$ is CP-divisible. We show that the decoherence function is decreasing in time. By Theorem \ref{thm_markov} the function $Q_{t,s}$ is positive definite. It is also hermitian. In particular, \eqref{posdefkernel} for $n=2$ gives
\begin{equation*}
\sum_{i,j=1}^{2} \overline{\xi_i} \xi_j Q_{t,s}(\vx_i,\vx_j) = 
\begin{pmatrix}
\overline{\xi_1} & \overline{\xi_2}
\end{pmatrix}
\begin{pmatrix}
Q_{t,s}(\vx_1,\vx_1) & Q_{t,s}(\vx_1,\vx_2) \\
\overline{Q_{t,s}(\vx_1,\vx_2)} & Q_{t,s}(\vx_2,\vx_2) 
\end{pmatrix}
\begin{pmatrix}
\xi_1 \\
\xi_2
\end{pmatrix} \geq 0
\end{equation*}
for any two-dimensional complex vector $(\xi_1,\xi_2)$. This means that for any pair of points $(\vx_1,\vx_2)$, the matrix in the above expression is positive definite. In particular, its determinant is positive,
\begin{equation*}
     Q_{t,s}(\vx_1,\vx_1)Q_{t,s}(\vx_2,\vx_2) - |Q_{t,s}(\vx_1,\vx_2) |^2 \geq 0.
\end{equation*}
Since $Q_{t,s}(\vx_1,\vx_1)=Q_{t,s}(\vx_2,\vx_2)=1$ we have $|Q_{t,s}(\vx_1,\vx_2)|\leq 1$ and therefore $|D_t(\vx_1,\vx_2)| \leq  |D_s(\vx_1,\vx_2)|$.
\hfill \qed
\bigskip

{\bf Data availability.} There is no data used in this work.
\bigskip

{\bf Competing interests.} 
The authors have no competing interests to declare that are relevant to the content of this article.
\bigskip

{\bf Acknowledgements.} M.M. thanks the Natural Sciences and Engineering Research Council of Canada (NSERC) for support through a Discovery Grant, and the Universit\'e C\^ote d'Azur for support and hospitality. 
S.M.~received financial support under the Horizon Europe research and innovation programme through the MSCA project ConNEqtions, n.~101056638, and the ERC StG MaTCh, grant agreement n.~101117299. S.M. also gratefully acknowledges funding from the Italian Ministry of University and Research and Next Generation EU through the PRIN 2022 project ONES, CUP:D53C24003430001. Furthermore, S.M. acknowledges the affiliation with Universit\'e C\^ote d'Azur during the initial stage of this work and he thanks Memorial University of Newfoundland for support and hospitality during the final stage of this work. The work of S.M. was performed under the auspices of GNFM-INDAM.

\end{document}